\documentclass[11pt,a4paper]{article}
\pdfoutput=1
\usepackage[margin=1in]{geometry}

\usepackage[utf8]{inputenc}
\usepackage{amsmath,amsfonts,amsthm}
\usepackage{booktabs,url,xspace,microtype,eucal,hyperref,xcolor,wrapfig,multirow}
\usepackage[inline]{enumitem}
\usepackage[numbers,sort&compress]{natbib}
\setlist[itemize]{label=--}
\setlist[enumerate]{label=(\arabic*)}
\usepackage[export]{adjustbox}
\usepackage{subcaption}
\usepackage{pdflscape}

\newtheorem{theorem}{Theorem}
\newtheorem{claim}[theorem]{Claim}
\newtheorem{lemma}[theorem]{Lemma}

\theoremstyle{remark}
\newtheorem{remark}[theorem]{Remark}

\newcommand{\Reals}{\mathbb{R}}

\newcommand{\range}{D}

\newcommand{\norm}[1]{\lVert #1 \rVert}
\newcommand{\residual}{\operatorname{res}}

\title{Relaxed Scheduling for Scalable Belief Propagation}

\newenvironment{myabstract}
{\list{}{\listparindent 1.5em%
		\itemindent    \listparindent
		\leftmargin    1cm
		\rightmargin   1cm
		\parsep        0pt}%
	\item\relax}
{\endlist}

\newenvironment{mycover}
{\list{}{\listparindent 0pt
		\itemindent    \listparindent
		\leftmargin    1cm
		\rightmargin   1cm
		\parsep        0pt}%
	\raggedright
	\item\relax}
{\endlist}

\newcommand{\myemail}[1]{\,$\cdot$\, {\small #1}}
\newcommand{\myaff}[1]{\,$\cdot$\, {\small #1}\par\medskip}

\begin{document}

\begin{mycover}
	{\huge\bfseries\boldmath Relaxed Scheduling for \\ Scalable Belief Propagation \par}
	\bigskip
	\bigskip
    \textbf{Vitaly Aksenov}
    \myemail{aksenov.vitaly@gmail.com}
    \myaff{ITMO University}    

    \textbf{Dan Alistarh}
    \myemail{dan.alistarh@ist.ac.at}
    \myaff{IST Austria}
    
    \textbf{Janne H.\ Korhonen}
    \myemail{janne.korhonen@ist.ac.at}
    \myaff{IST Austria}
\end{mycover}

\medskip
\begin{myabstract}
  \noindent\textbf{Abstract.}
    The ability to leverage large-scale hardware parallelism has been one of the key enablers of the accelerated recent progress in machine learning. Consequently, there has been considerable effort invested into developing efficient parallel variants of classic machine learning algorithms. However, despite the wealth of knowledge on parallelization, some classic machine learning algorithms often prove hard to parallelize efficiently while maintaining convergence.

    In this paper, we focus on efficient parallel algorithms for the key machine learning task of inference on graphical models, in particular on the fundamental belief propagation algorithm. We address the challenge of efficiently parallelizing this classic paradigm by showing how to leverage scalable relaxed schedulers in this context. We present an extensive empirical study, showing that our approach outperforms previous parallel belief propagation implementations both in terms of scalability and in terms of wall-clock convergence time, on a range of practical applications.
\end{myabstract}

\thispagestyle{empty}
\setcounter{page}{0}

\section{Introduction}

\emph{Hardware parallelism} has been a key computational enabler for recent advances in machine learning, as it provides a way to reduce the processing time for the  ever-increasing quantities of data required for training accurate models.
Consequently, there has been considerable effort invested into developing  efficient parallel variants of classic machine learning algorithms, e.g.~\citep{recht2011hogwild, PS, liu2015asynchronous, low2014graphlab, pmlr-v5-gonzalez09a}.

In this paper, we will focus on efficient parallel algorithms for the fundamental task of \emph{inference on graphical models}.  
The inference task in graphical models takes the form of \emph{marginalisation}: we are given observations for a subset of the random variables, and the task is to compute the conditional distribution of one or a few variables of interest. The marginalization problem is known to be computationally intractable in general~\citep{DAGUM1993141,ROTH1996273,COOPER1990393}, but inexact heuristics are well-studied for  practical inference tasks.

One popular heuristic for inference on graphical models is \emph{belief propagation}~\citep{bp}, inspired by the exact dynamic programming algorithm for marginalization on trees. While belief propagation has no general approximation or even convergence guarantees, it has proven empirically successful in inference tasks, in particular in the context of decoding low-density parity check codes~\citep{casado2007informed}. However, it remains poorly understood how to properly parallelize belief propagation.

\paragraph{Parallelizing belief propagation.} To illustrate the challenges of parallelizing belief propagation, we will next give a simplified overview of the belief propagation algorithm, and refer the reader to Section~\ref{section:preliminaries} for full details. Belief propagation can be seen as a \emph{message passing} or a \emph{weight update} algorithm. In brief, belief propagation operates over the underlying graph $G = (V,E)$ of the graphical model, maintaining a vector of real numbers called a \emph{message} $\mu_{i \to j}$ for each ordered pair $(i,j)$ corresponding to an edge $\{ i, j \} \in E$ (Fig.~\ref{fig:bp-messages}). The core of the algorithm is the \emph{message update rule} which specifies how to update an outgoing message $\mu_{i \to j}$ at node $i$ based on the \emph{other} incoming messages at node $i$; for the purposes of the present discussion, it is sufficient to view this as black box function $f$ over these other messages, leading to the update rule
\begin{equation}\label{eq:bp-abstract}
\mu_{i \to j} \gets f\bigl( \{ \mu_{k \to i}  \colon k \in N(i) \setminus \{ j \} \} \bigr)\,.
\end{equation}
This update rule is applied to messages until the values of messages have converged to a stable solution, at which point the algorithm is said to have terminated.

\begin{wrapfigure}{rt}{0.4\linewidth}
\center
\includegraphics[width=0.75\linewidth]{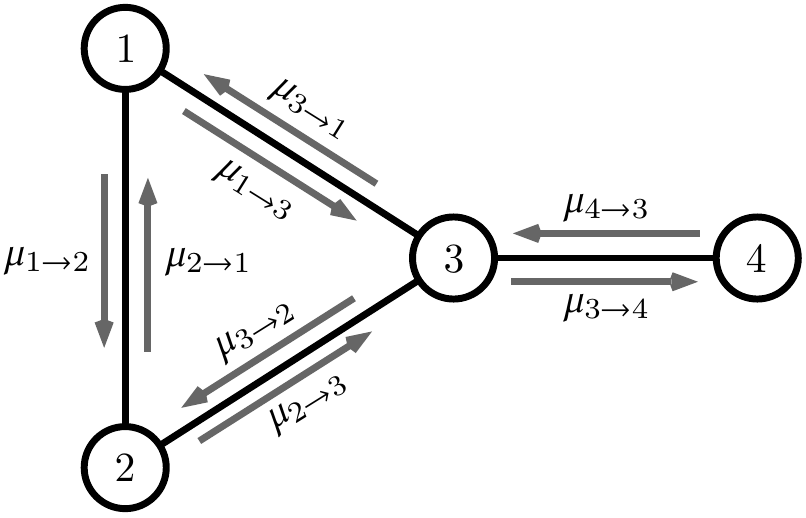}
\caption{\boldmath State of the belief propagation algorithm consist of two directed messages for each edge. %
}\label{fig:bp-messages}
\end{wrapfigure}

Importantly, the message update rule does not specify \emph{in which order} messages should be updated. The standard solution, called \emph{synchronous belief propagation}, is to update all the message simultaneously. That is, in each global round $t = 1, 2, 3, \dotsc$, given message values $\mu^t_{i \to j}$ for all pairs $(i,j)$, the new values $\mu^{t+1}_{i \to j}$ are computed as 
\begin{equation*}
\mu^{t+1}_{i \to j} \gets f\bigl( \{ \mu^{t}_{k \to i}  \colon k \in N(i) \setminus \{ j \} \} \bigr) 
\end{equation*}
However, there is  evidence that updating messages \emph{one at a time} leads to faster and more reliable convergence~\citep{elidan2006residual}; 
in particular, various proposed \emph{priority-based schedules}---schedules that try to prioritize message updates that would make `more progress'---have proven empirically to converge with much fewer message updates than the synchronous schedule~\citep{elidan2006residual,knoll2015message, Sutton:2007:IDS:3020488.3020534}.

Having to execute updates in a strict priority order poses a challenge for efficient \emph{parallel} implementations of belief propagation: while the synchronous schedule is naturally parallelizable, as all message updates can be done independently, the more efficient priority-based schedules are inherently sequential and thus seem difficult to parallelize. Accordingly, existing work on efficient parallel belief propagation has focused on designing custom schedules that try to import some features from the priority-based schedules while maintaining a degree of parallelism~\citep{pmlr-v5-gonzalez09a,merwe2019}.

\paragraph{Our contributions.}
In this work, we address this challenge by studying how to efficiently parallelize any priority-based schedule for belief propagation. The key idea is that we can \emph{relax} the priority-based schedules by allowing limited out-of-order execution, concretely implemented using a \emph{relaxed scheduler}, as we will explain next.

Consider a belief propagation algorithm that schedules the message updates according to a priority function $r$ by always updating the message $\mu_{i\to j}$ with the highest priority $r(\mu_{i \to j})$ next; this framework captures existing priority-based schedules such as residual belief propagation~\citep{elidan2006residual} and its variants~\citep{knoll2015message, Sutton:2007:IDS:3020488.3020534}. Concretely, an iterative centralized version of this algorithm can be implemented by storing the messages in a priority queue $Q$, and iterating the following procedure: 
\begin{enumerate}[label=(\arabic*)]
    \item Pop the top element for $Q$ to obtain the message $\mu_{i \to j}$ with highest priority $r(\mu_{i \to j})$.
    \item Update message $\mu_{i \to j}$ following (\ref{eq:bp-abstract}).
    \item Update the priorities in $Q$ for messages affected by the update.
\end{enumerate}
This template does not easily lend itself to efficient parallelization, as the priority queue $Q$ becomes a contention bottleneck. 
Previous work, e.g.~\citep{pmlr-v5-gonzalez09a, merwe2019} investigated various heuristics for the parallel scheduling of updates in belief propagation, trading off increased parallelism with additional work in processing messages or even potential loss of convergence. 

In this paper, we investigate an alternative approach, replacing the priority queue $Q$ with a \emph{relaxed priority queue (scheduler)} to obtain a efficient parallel version of the above template. A \emph{relaxed scheduler}~\cite{AKLN17, alistarhbkln18, alistarhBKN18, alistarhnk2019} is similar to a priority queue, but instead of guaranteeing that the \emph{top} element is always returned first, it only guarantees to return \emph{one of the top $q$ elements by priority} , where $q$ is a variable parameter. 
Relaxed schedulers are popular in the context of parallel graph processing, e.g.~\citep{gonzalez2012powergraph, Nguyen13}, and can induce non-trivial trade-offs between the degree of relaxation and the scalability of the underlying implementation, e.g.~\citep{alistarhbkln18, alistarhnk2019}. 

For belief propagation, relaxed schedulers induce a \emph{relaxed priority-based scheduling} of the messages, roughly following the original schedule but allowing for message updates to be performed out of order, with guarantees on the maximum degree of priority inversion. 
We investigate  the scalability-convergence trade-off between the \emph{degree of relaxation} in the scheduler, and the \emph{convergence behavior} of the underlying algorithm, both theoretically and practically. In particular: 

\begin{itemize}[leftmargin=4.5mm]
    \item We present a general scheme for parallelizing belief propagation schedules using relaxed schedulers with theoretical guarantees. While relaxed schedulers have been applied in other settings, and relaxed scheduling has been touched upon in belief propagation scheduling~\cite{pmlr-v5-gonzalez09a}, no systematic study on relaxed belief propagation scheduling has been performed in prior work.
    \item We provide a theoretical analysis on the effects of relaxed scheduling for belief propagation on trees. We exhibit both positive results--instance classes where relaxation overhead is negligible--and negative results, i.e., worst-case instances where relaxation causes significant wasted work. 
\end{itemize}

\paragraph{Experimental evaluation.} We implement our relaxed priority-based scheduling framework with a \emph{Multiqueue} data structure~\citep{rihani2015brief} and instantiate it with several known priority-based schedules. In the benchmarks, we show that this framework gives state-of-the-art parallel scalability on a wide variety of Markov random field models. As expected, the relaxed priority-based schedules require slightly more message updates than their exact counterparts, but this performance overhead is offset by their better scalability. In particular, we highlight the fact that the relaxed version of the popular residual belief propagation algorithm has state-of-the art parallel scaling while requiring no tuning parameters, making it an attractive practical solution for belief propagation. This finding is illustrated in Figure~\ref{figure:bars-intro}, and further substantiated in Section~\ref{section:experiments}.

\begin{figure*}
\begin{center}
\includegraphics[width=0.45\textwidth]{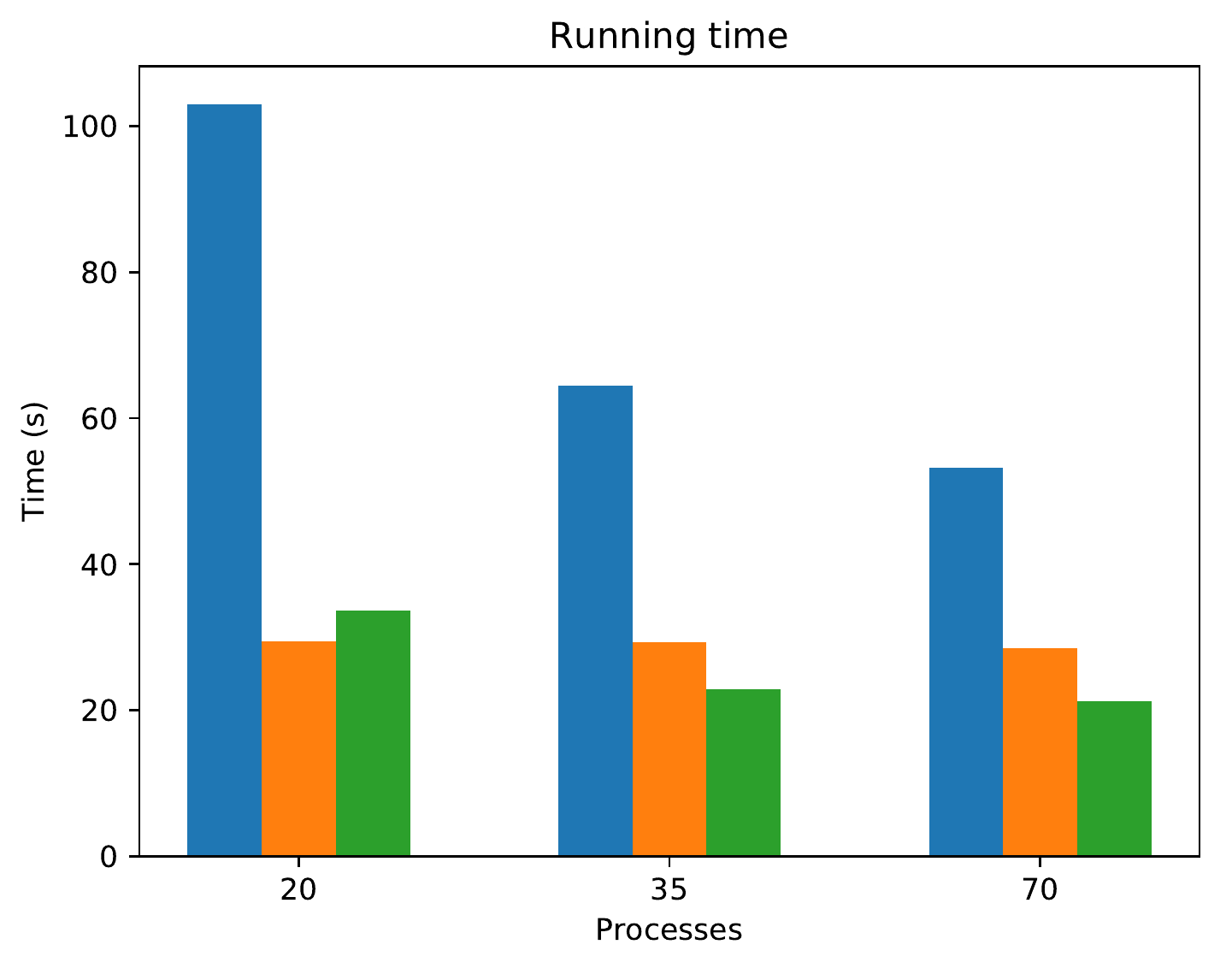}
\includegraphics[width=0.443\textwidth]{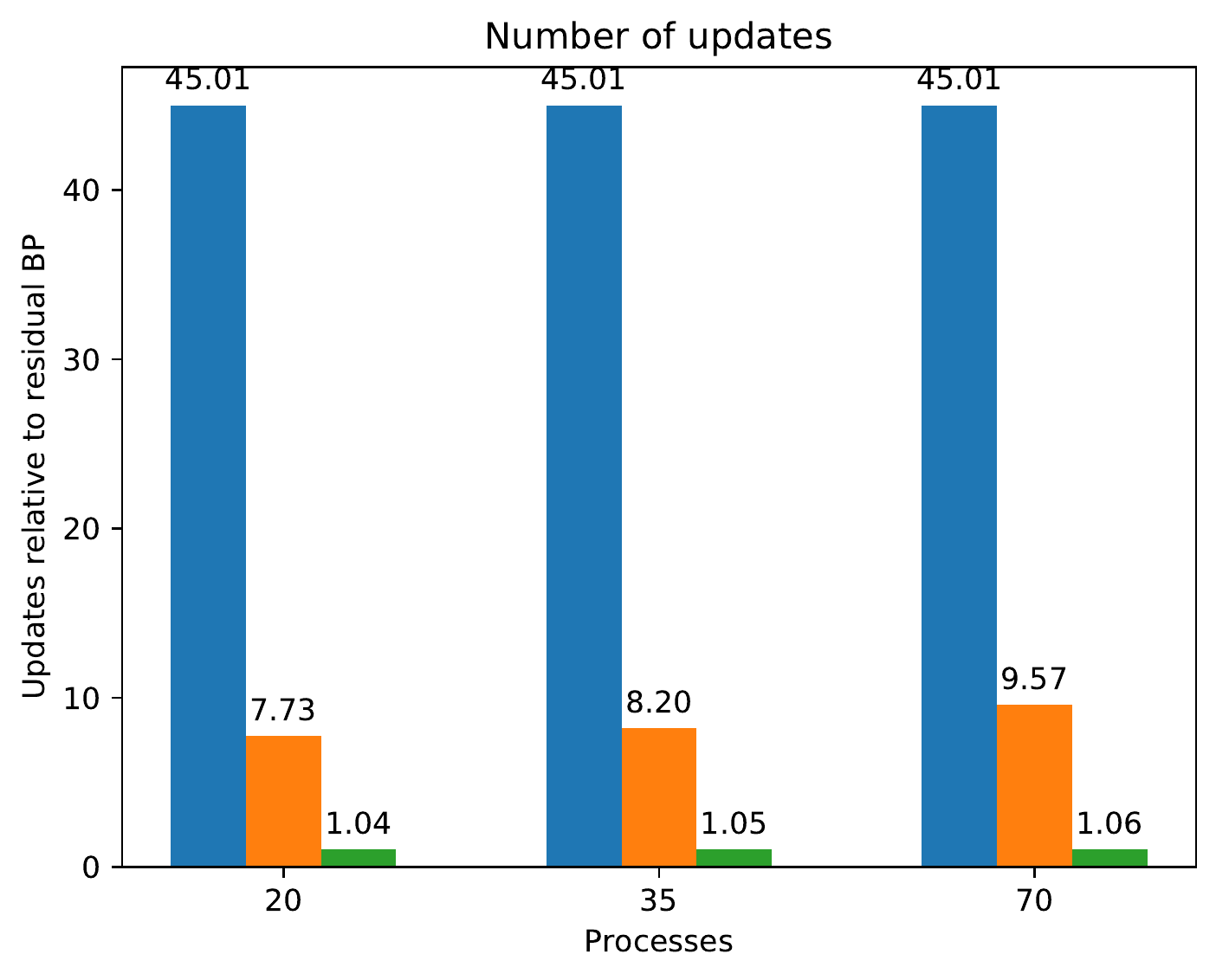}
\caption{Running time (\textbf{left}) and number of updates until convergence (\textbf{right}) on a $1000 \times 1000$ Ising grid model on $p \in \{20, 35, 70 \}$ processes. Included algorithms are synchronous belief propagation, residual splash belief propagation of Gonzalez et al.~\cite{pmlr-v5-gonzalez09a} with splash size 10, and our relaxed residual belief propagation. The number of updates is relative to sequential residual belief propagation.}\label{figure:bars-intro}
\end{center}
\end{figure*}

\section{Preliminaries and Related Work}\label{section:preliminaries}

\subsection{Belief Propagation}

We consider marginalization in \emph{pairwise Markov random fields}; one can equivalently consider factor graphs or Bayesian networks~\citep{yedidia2003}. A pairwise Markov random field is defined by a set of random variables $X_1, X_2, \dotsc, X_n$, a graph $G = (V,E)$ with $V = \{ 1, 2, \dotsc, n \}$, and a set of \emph{factors}
\begin{align*}
    \psi_{i} & \colon \range_i \to \Reals^+ && \text{for $i \in V$,}\\
    \psi_{ij} & \colon \range_i \times \range_j \to \Reals^+ && \text{for $\{i, j\} \in E$,}
\end{align*}
where $\range_i$ denotes the domain of random variable $X_i$. The edge factors $\psi_{ij}$ represent the dependencies between the random variables, and the node factors $\psi_{i}$ represent a priori information about the individual random variables; the Markov random field defines a joint probability distribution on $X = (X_1, X_2, \dotsc, X_n)$ as
\[ \Pr\bigl[ X = x \bigr] \propto \prod_{i} \psi_i(x_i) \prod_{ij} \psi_{ij}(x_i, x_j) \,,\]
where the `proportional to' notation $\propto$ hides the normalization constant applied to the right-hand side to obtain a probability distribution.
The marginalization problem is to compute the probabilities $\Pr[ X_i = x ]$ for a specified subset of variables; for convenience, we assume that any  observations regarding the values of other  variables are encoded in the node factor functions $\psi_i$.  

Belief propagation is a message-passing algorithm; for each ordered pair $(i,j)$ such that $\{ i,j \} \in E$, we maintain a \emph{message} $\mu_{i \to j} \colon \range_j \to \Reals$, and the algorithm iteratively updates these messages until the values (approximately) converge to a fixed point. On Markov random fields, the message update rule gives the new value of message $\mu_{i \to j}$ as a function of the old messages directed to node $i$ by
\begin{equation}\label{eq:bp}
\mu_{i \to j}(x_j) \propto \sum_{x_i \in \range_i} \psi_i(x_i)\psi_{ij}(x_i,x_j) \prod_{k \in N(i)\setminus \{ j \}} \mu_{k \to i} (x_i)\,, 
\end{equation}
where $N(j)$ denotes the neighbors of node $j$ in the graph $G$. Once the algorithm has converged, the marginals are estimated as
$\Pr[ X_i = x_i ] \propto \psi_i(x_i) \prod_{j \in N(i)} \mu_{j \to i}(x_i)$.

The update rule (\ref{eq:bp}) can be applied in arbitrary order. The standard \emph{synchronous belief propagation} updates all the message simultaneously; in each global round $t = 1, 2, 3, \dotsc$, given message values $\mu^t_{i \to j}$ for all pairs $\{i,j\} \in E$, the new values $\mu^{t+1}_{i \to j}$ are computed as 
\begin{equation*}
\mu^{t+1}_{i \to j}(x_j) \propto \sum_{x_i \in \range_i} \psi_i(x_i)\psi_{ij}(x_i,x_j) \prod_{k \in N(i)\setminus \{ j \}} \mu^{t}_{k \to i} (x_i)\,. 
\end{equation*}

\subsection{Asynchronous belief propagation}
Starting with Elidan et al.~\citep{elidan2006residual}, there has been a line of research arguing that \emph{asynchronous} or \emph{iterative} schedules for belief propagation tend to converge more reliably and with fewer message updates that the synchronous schedule. In particular, the practical work has focused on developing schedules that attempt to iteratively perform `the most useful' update at each step; the most prominent of these algorithms is the \emph{residual belief propagation} of Elidan et al.~\citep{elidan2006residual}, with other proposals aiming to address the shortcomings of residual belief propagation in various cases.

\paragraph{Residual belief propagation.} Given a current state of messages, let $\mu'_{i \to j}$ denote the message we would obtain by applying the message update rule (\ref{eq:bp}) to message $\mu_{i \to j}$. In residual belief propagation, the priority of a message is given by the \emph{residual} $\residual(\mu_{i \to j})$ of a message $\mu_{i \to j}$, defined as
\begin{equation}\label{app:eq:residual}
\residual(\mu_{i \to j}) = \norm{\mu'_{i \to j} - \mu_{i \to j}}\,, 
\end{equation}
where $\norm{\cdot}$ is an arbitrary norm; in this work, we assume $L^2$ norm is used unless otherwise specified. That is, the residual of a message corresponds to amount of change that would happen if message $\mu_{i \to j}$ would be updated. Note that this means that  residual belief propagation performs \emph{lookahead}, that is, the algorithm precomputes the future updates before applying them to the state of the algorithm.

\paragraph{Weight decay belief propagation.} \emph{Weight decay belief propagation} of~\citep{knoll2015message} is a variant of residual belief propagation that penalizes message priorities for repeated updates. That is, let $m(\mu_{i \to j})$ denote how many times message $\mu_{i \to j}$ has been updated by the algorithm, and let $\residual(\mu_{i \to j})$ denote the residual of a message as above. The priority function of weight decay belief propagation is
\[ r(\mu_{i \to j}) = \frac{\residual(\mu_{i \to j})}{m(\mu_{i \to j})}\,.\]
The motivation behind this weight decay scheme is that empirical observations suggest that one possible failure mode of residual belief propagation is getting stuck in cycles with large residuals; the weight decay prioritizes other edges in cases where this happens.

\paragraph{Residual without lookahead.} Another variant of residual belief propagation is the lookahead-avoiding belief propagation of~\citep{Sutton:2007:IDS:3020488.3020534}. As the name implies, this algorithm does not perform the exact residual computation using (\ref{app:eq:residual}), but instead approximates the residuals indirectly, with the aim of reducing the computational cost of priority updates.

Informally, the basic idea is that for each message $\mu_{i \to j}$, we track the amount other incoming messages at node $i$ have changed since the last update of $\mu_{i \to j}$, and use this to define the priority of updating $\mu_{i \to j}$. The actual approximation in the algorithm uses a slightly different notion of residual from (\ref{app:eq:residual}), so we refer to~\citep{Sutton:2007:IDS:3020488.3020534} for full details.

\subsection{Parallel belief propagation}\label{app:parallel-bp}

The question of parallelizing belief propagation is not fully understood. The synchronous schedule is trivially parallelizable by performing updates within each round in parallel, but the improved convergence properties of iterative schedules cannot easily be translated to parallel setting. Recent proposals aim to bridge this gap in an ad-hoc manner by designing custom algorithms for specific parallel settings. We discuss the most relevant ones below, omitting ones that apply strictly to a multi-machine distributed setting (e.g. Gonzalez et al.~\cite{10.5555/1795114.1795139}.)

\paragraph{Residual splash.} The \emph{residual splash} belief propagation~\citep{pmlr-v5-gonzalez09a} is a vertex-based algorithm inspired by residual belief propagation. The residual splash algorithm was initially designed for MapReduce computation, and it aims to have larger individual tasks while retaining a similar structure to residual belief propagation.

Specifically, the residual splash algorithm works by defining a priority function over nodes of the Markov random field, and selecting the next node to process in a strict priority order. For the selected node, the algorithm performs a \emph{splash} operation that propagates information within distance \emph{H} in the graph; in practice, this results in threads performing larger individual tasks at once, offsetting the cost of accessing the strict scheduler.

In detail, the priority of for nodes is given by the \emph{node residual}, defined as
\[ \residual(i) = \max_{j \in N(i)} \residual(\mu_{j \to i})\,.\]
Given a \emph{depth parameter} $H$, the splash operation at node $i$ is defined by following sequence of message updates:
\begin{enumerate}
    \item Construct a BFS tree $T$ of depth $H$ rooted at node $i$.
    \item In the reverse BFS order on $T$---starting from leaves---process all nodes in $T$, updating all outgoing messages for each node processed.
    \item Repeat the previous step in BFS order, i.e., starting from the root.
\end{enumerate}
In other words, this process gather all available information at radius $H$ from the selected node, and propagates it to all nodes within the radius. 

\paragraph{Mixed synchronous and priority-based belief propagation.} Mixed strategies for belief propagation scheduling have also been proposed, by Van der Merve et al.~\citep{merwe2019} for GPUs and by Yin and Gao~\citep{10.1145/2661829.2662081} for distributed setting. These proposals use residuals or other similar priority functions to select a set of high-score  messages to update at each step, and then perform those updates as in the synchronous schedule. While these algorithms work well in their original contexts, where relaxed schedulers and other concurrent data structures cannot be utilized, this strategy is not efficient on shared-memory parallel setting on CPUs, as seen in the experiments.


\section{Relaxed Priority-based Belief Propagation}\label{section:relaxation}


In this section, we describe our framework for parallelizing belief propagation schedules via relaxed schedulers. The main idea of the framework follows the description given in the introduction; however, we generalize it to capture schedules that do not use individual messages as elementary tasks, e.g. residual belief propagation~\citep{pmlr-v5-gonzalez09a}.

\subsection{Relaxed Scheduling for Iterative Algorithms}\label{sec:relaxed-schedulers}

Relaxed schedulers are a basic tool to parallelize iterative algorithms, used in the context of large-scale graph processing~\citep{gonzalez2012powergraph, Nguyen13, blelloch2016parallelism,  dhulipala17julienne, dhulipala2018theoretically}. 
At a high level, such iterative algorithms can be split into \emph{tasks}, each corresponding to a local operation involving, e.g., a vertex and its edges. A standard example is Dijkstra's classic shortest-paths algorithm, where each task updates the distance between a vertex and the source, as well as the distances of the vertex's neighbours. In many algorithms, tasks have a natural priority order---in Dijkstra's, the top task corresponds to the vertex of minimum distance from the source. 
Many graph algorithms share this structure~\cite{shun13priority, alistarhBKN18}, and can be mapped onto the priority-queue pattern described in the introduction. 
However, due to contention on the priority queue, implementing this pattern in practice can negate the benefits of parallelism~\cite{Nguyen13}. 

\paragraph{Relaxed scheduler definition.} 
A natural idea is to downgrade the guarantees of the perfect priority queue, to allow for more parallelism. 
Relaxed schedulers~\citep{alistarhbkln18} formalize this relaxation as follows. 
We are given a parameter $q$, the degree of relaxation of the scheduler. 
The $q$-relaxed scheduler is a \emph{sequential} object supporting \texttt{Insert (<key, priority>)}, \texttt{IncreaseKey (<key, priority>)}, with the usual semantics, and an \texttt{ApproxDeleteMin()} operations, ensuring:
\begin{enumerate}
    \item \textbf{Rank Bound.}  \texttt{ApproxDeleteMin()} returns one of the top $q$ elements in priority order. 
    \item  \textbf{\boldmath $q$-fairness.} We say that a \emph{priority inversion} on element \texttt{<key, priority>} is the event that \texttt{ApproxDeleteMin()} returns a key with a \emph{lower} priority while \texttt{<key, priority>} is in the queue. The $q$-fairness condition requires that any element can experience at most $q$ priority inversions before it must be returned. 
\end{enumerate}

Relaxed schedulers are quite popular in practice, as several efficient implementations have been proposed and applied~\citep{LotanShavit, Basin11, klsm, SprayList, Haas, Nguyen13, MQ, AKLN17, sagonas2017contention}, with state-of-the-art results in the graph processing domain~\citep{Nguyen13, gonzalez2012powergraph, Swarm}. 
A parallel line of work has attempted to provide guarantees on the amount of relaxation in individual schedulers~\citep{AKLN17, alistarhBKN18, Rukundo2019}, 
as well as the impact of using relaxed scheduling on existing iterative algorithms~\citep{alistarhbkln18, alistarhnk2019}.
Here, we employ the modeling of relaxed schedulers used in e.g.~\citep{alistarhBKN18, alistarhnk2019} for graph algorithms, but apply it to inference on graphical models.

\subsection{Relaxed Priority-based Belief Propagation}

Given a Markov random field, a priority-based schedule for BP is defined by a set of \emph{tasks} $T_1, T_2, \dotsc, T_K$, each corresponding to a sequence of edge updates, and a priority function $r$ that assigns a priority $r(T_i)$ to a task based on the current state of the messages as well as possible auxiliary information maintained separately. As discussed in the introduction, a non-relaxed algorithm can store all tasks in a priority queue, iteratively retrieve the highest-priority task, perform all its message updates, and update priorities accordingly.

The relaxed variant works in exactly the same way, but assuming a \emph{$q$-relaxed} priority scheduler $Q_q$. More precisely, the following steps are repeated until a fixed convergence criterion is reached:
\begin{enumerate}[label=(\arabic*)]
    \item $T_i \gets Q_q.\texttt{ApproxDeleteMin}()$ selects a task $T_i$ among the $q$ of highest priority in $Q_q$.
    \item Perform all message updates specified by the task~$T_i$.
    \item Update the priorities for all tasks.
\end{enumerate}
Note that tasks can be executed multiple times in this framework. In particular, we assume that the priority $r(T_i)$ of a task $T_i$ can only remain the same or increase when other tasks are executed, and the only point where the priority decreases is when the task is actually executed.

\subsection{Concurrent Implementation} 
\label{sec:algo}

The sequential version of a priority-based schedule for belief propagation can be implemented using a priority queue $Q$. 
One could map the sequential pattern directly to a parallel setting, by replacing the sequential priority queue with a linearizable concurrent one. 
However, this may not be the best option, for two reasons. 
First, it is challenging to build \emph{scalable} exact priority queues~\citep{lenharth2015priority}---the data structure is inherently contended, which leads to poor cache behavior and poor performance. 
Second, in this context, linearizability only gives the illusion of atomicity with respect to task message updates: 
the data structure only ensures that the \emph{task removal} is atomic, whereas the actual message updates which are part of the task are not usually performed atomically together with the removal. 

\paragraph{Multiqueue.} 
For the aforementioned reasons, in our framework, we use a \emph{relaxed} priority scheduler, specifically a scalable approximate priority queue called the Multiqueue~\citep{rihani2015brief, AKLN17}. 
As the name suggests, the Multiqueue is composed of $m$ independent \emph{exact} priority queues. 
To \texttt{Insert} an element, a thread picks one of the exact priority queues uniformly at random, and inserts into it. 
To perform \texttt{ApproxDeleteMin()}, the thread picks \emph{two} of the priority queues uniformly at random, and removes the \emph{higher priority} element among their two top elements. 
Although very simple, this strategy has been shown to have strong probabilistic rank and fairness guarantees:

\begin{theorem}[\cite{AKLN17, alistarhbkln18}]
    A Multiqueue formed of $p \geq 3$ priority queues ensures the rank and fairness guarantees with parameter $q = O( p \log p )$, with high probability. 
\end{theorem}

\paragraph{Implementation details.}
For our purposes, we assume that each thread $i$ has one or a few local concurrent priority queues, used to store pointers to BP-specific tasks (e.g. messages), which are prioritized by an algorithm-specific function, e.g. the residual values for residual BP. 
We store additional metadata as required by the algorithm and the graphical model. 
(In our experiments, we use binary heaps for these priority queues, protected by locks.) 
To process a new task, the thread selects two among all the priority queues uniformly at random, and accesses the task/message from the queue whose top element has higher priority. The task is marked as \emph{in-process} so it cannot be processed concurrently by some other thread. 
The thread then proceeds to perform the metadata updates required by the underlying variant of belief propagation, e.g., updating the message and the priorities of messages from the corresponding node. 
The termination condition, e.g., the magnitude of the largest update, is checked periodically. 


\section{Dynamics of Relaxed Belief Propagation on Trees}\label{section:theory}


As we will see in Section~\ref{section:experiments}, the relaxed priority-based belief propagation schedules yield fast converge times on a wide variety of Markov random fields; specifically, the number of message updates is roughly the same as for the non-relaxed version, while the running times are lower. The complementary theoretical question we examine here is whether we can give analytical bounds how much extra work---in terms of additional message updates---the relaxation incurs in the worst-case. 

Unfortunately, the dynamics of even synchronous belief propagation are poorly understood on general graphs, and no priority-based algorithms provide general guarantees on the convergence time. As such, we can only hope to gain some limited understanding on why relaxation retains the fast convergence properties of the exact priority-based schedules. Here, we present some theoretical evidence suggesting that as long as a schedule does not impose long dependency chains in the sequence of updates, relaxation incurs low overhead, but also that simple (non-loopy) worst-case instances exist. 

\paragraph{Analytical model.}
For analysis of the relaxed priority-based belief propagation, we consider the formal model introduced by~\citep{alistarhnk2019, alistarhBKN18} to analyze performance of iterative algorithms under relaxed schedulers. Specifically, we model a relaxed scheduler $Q_q$ as a data structure which stores pairs corresponding to tasks and their priorities, with the operational semantics given in Section~\ref{section:relaxation}. In particular, there exists a parameter $q$ such that each \texttt{ApproxDeleteMin} returns one of the $q$ highest priority tasks in $Q_q$, and if a task $T$ becomes the highest priority task in $Q_q$ at some point during the execution, then one of the next $q$ \texttt{ApproxDeleteMin} operations returns $T$. (By~\citep{AKLN17, alistarhbkln18}, our practical implementation will satisfy these conditions with parameter $q = O(p \log p)$ w.h.p., where $p$ is the number of concurrent threads.)
Other than satisfying these properties, we assume that the behavior of $Q_q$ can be adversarial, or randomized.


We model the behavior of relaxed priority-based belief propagation by investigating the number of message updates needed for convergence when the algorithm is executed \emph{sequentially} using a relaxed scheduler $Q_q$ satisfying the above constraints. 
This analysis reduces to a sequential game between the algorithm, which queries $Q_q$ for tasks/messages, and the scheduler, which returns messages in possibly arbitrary fashion. One may think of the relaxed sequential execution as a form of linearization for the actual parallel execution---reference~\citep{alistarhbkln18} formalizes this intuition. Please see the discussion at the end of this section for a practical interpretation.

\paragraph{Relaxed belief propagation on trees.} 
We now consider the behavior of relaxed residual belief propagation schedules on \emph{trees with a single source}. The setting is similar to the analysis of residual splash of Gonzalez et al.~\citep{pmlr-v5-gonzalez09a}. 
Specifically, we assume that the Markov random field and the initialization of the algorithm satisfies (1) The graph $G = (V,E)$ is a tree with a specified root $r$; and (2) The factors of the Markov random field and the initial messages are such that the residuals are zero for all messages other than the outgoing messages from the root, i.e., $\residual(\mu_{i \to j}) = 0$ if $i \ne r$.

These conditions mean that residual belief propagation will start from the root, and propagate the messages down the trees until propagation reaches all leaves. In particular, residual belief propagation without relaxation will perform $n-1$ message updates before convergence, updating each message away from root once.
While this setting is restrictive, it does model practical instances where the MRF has locally tree-like structure, such as LDPC codes (see Section~\ref{section:experiments}).

To characterize the dynamics on relaxed residual belief propagation on trees with a single source, we observe that the algorithm can make two types of message updates:
\begin{itemize}[leftmargin=4.5mm]
    \item Updating a message with zero residual, in which case nothing happens (\emph{a wasted update}). This happens if the scheduler relaxes past the range of messages with non-zero residual. 
    \item Updating a message $\mu_{i \to j}$ with non-zero residual, in which case the residual of $\mu_{i \to j}$ goes down to zero, and the messages $\mu_{j \to k}$ for the children $k$ of $j$ may change their residuals to non-zero values (\emph{a useful update}).
\end{itemize}
It follows that each edge will get updated only once with non-zero residual. At any point of time during the execution of the algorithm, we say that the \emph{frontier} is the set of messages with non-zero residual, and use $F(t)$ to denote the size of the frontier at time step $t$.

To see how the size of the frontier relates to the number message updates in relaxed residual belief propagation, observe that after a useful update, we have one of the following cases:
\begin{itemize}[leftmargin=4.5mm]
    \item If $F(t) \ge q$, then the next \texttt{ApproxDeleteMin}() operation to $Q_q$ will give an edge with non-zero residual, resulting in a useful update.
    \item If $F(t) < q$, then in the worst case we need $q$ \texttt{ApproxDeleteMin}() operations until we perform a useful update.
\end{itemize}
Our main analytic result bounds the worst-case work incurred by relaxation in two concrete cases.

\begin{lemma}
\label{lem:cases}
    Assume a $q$-relaxed scheduler $Q_q$ for belief propagation in the tree model defined above. 
    The total number of updates performed by relaxed residual BP can be bounded as follows:
    \begin{itemize}[leftmargin=4.5mm]
        \item \textbf{\emph{Good case: uniform expansion}.} If the tree model has identical and non-deterministic edge factors $\psi_{ij}$ with $\psi_{ij}(x_i,x_j) \ne 0$ for all $\{i, j\}$, then the total number of updates performed by relaxed residual BP is $n + O(Hq^2)$.
        
        \item \textbf{\emph{Bad case: long paths}.} There exists a tree instance with height $H = o(n)$ and an adversarial scheduler where relaxed residual belief propagation performs $\Omega(qn)$ message updates. 
    \end{itemize}
\end{lemma}

\begin{proof}

\emph{Good case: uniform expansion.}
As the first case, we consider the tree model in the case where the edge factors $\psi_{ij}$ are identical for all edges and not \emph{deterministic}, i.e. $\psi_{ij}(x_i,x_j) \ne 0$ for all $\{i, j\}$. Let us say that the \emph{level} of a message $\mu_{i \to j}$ is $\ell$ if the distance from $i$ to the root $r$ is $\ell$. The conditions we imposed our Markov random field, together with the update rule (\ref{eq:bp}), imply that the residuals of the messages are decreasing in the level $\ell$ of the message, and all messages on level $\ell$ will have the same residual when they are in the frontier. This means that residual schedule will prefer updating messages on lower levels first.

Now consider the progress of the relaxed residual belief propagation on this tree; let $H$ denote the height of the tree. Now assume that all messages on levels $0, 1, \dotsc, \ell - 1$ have had a useful update, and consider how many wasted updates we can make in the worst case before all messages on level $\ell$ have been processed. Let $f$ denote the number of message on level $\ell$ still in the frontier:
\begin{itemize}
    \item While $f \ge k$, there are at least $k$ messages of level $\ell$ on the frontier. Since they have the highest residual out of the messages in the frontier, each update is a useful update of a message on level $\ell$.
    \item When $f < k$, there can be updates that do not hit messages on level $\ell$, which can possibly be wasted updates. However, the highest-priority messages are still from level $\ell$, so every $k$th update will hit a message on level $\ell$ by the guarantees of the scheduler. Thus, in $(k-1)f = O(k^2)$ updates, all remaining messages on level $\ell$ have been processed.
\end{itemize}
Since there can be at most $n - 1$ useful updates, and the number of levels is $H - 1$, the total number of updates performed by relaxed residual belief propagation is $n + O(Hk^2)$.

\emph{Bad case: long paths.} A simple example where relaxed residual belief propagation performs poorly is a path. That is, if our underlying tree is a path of length $n$ with a root at one end, then relaxed residual belief propagation can perform $\Omega(kn)$ message updates in the worst case. However, the path has height $H = n$, so one might ask if there is a general upper bound of form $n + O(Hk^2)$ on trees without restricting the edge factors as in our previous example.

Unfortunately, turns out that without the restrictions above, we can construct examples of trees with height $H = o(n)$ where relaxed residual belief propagation still performs $\Omega(kn)$ message updates (see Figure~\ref{fig:hard-instance} for an illustration):
\begin{enumerate}[nosep]
    \item Start with a path of length $\sqrt{n}$, with a root at one end.
    \item Attach a new path of length $\sqrt{n}$ to each vertex.
    \item For each remaining degree-2 node in the graph, attach a single new node to it.
\end{enumerate}

\begin{figure}
\centering
\includegraphics[width=0.5\linewidth]{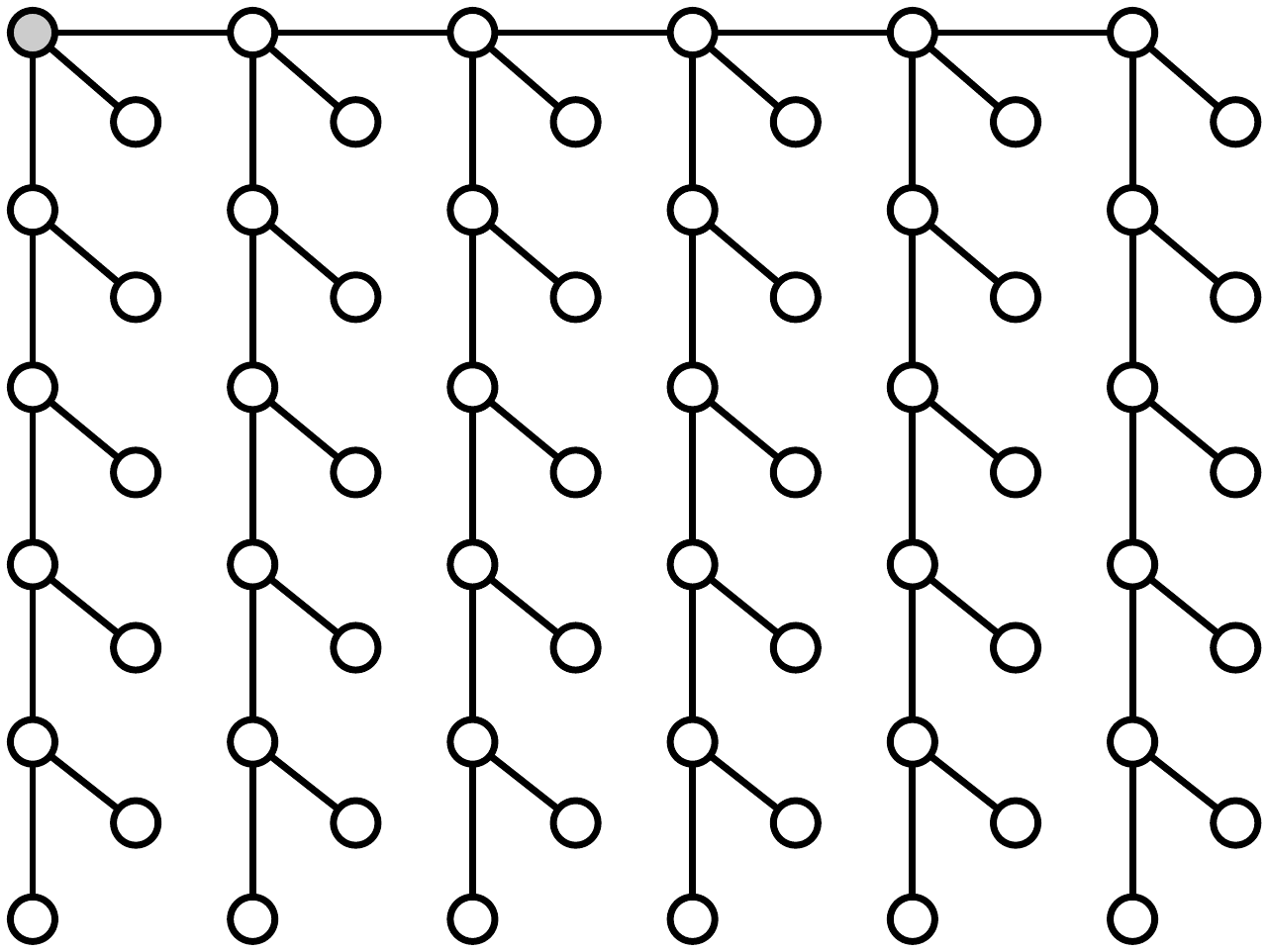}
\caption{\boldmath Example of the tree where relaxed residual belief propagation performs poorly.}\label{fig:hard-instance}    
\end{figure}

This construction results in a $3$-regular rooted tree with $\Theta(n)$ nodes and depth $H = O(\sqrt{n})$. Finally, we choose the edge factors so that residuals on the side paths are larger than the residuals on the main path, so residual belief propagation will prefer following the side paths first. 

One can now observe that under the adversarial model for the relaxed scheduler, the adversary can select the execution of the relaxed scheduler so that the frontier size never exceeds $4$. That is, adversary forces the algorithm to process the graph one side path at time, wasting $k-1$ steps between each useful update.

Finally, we note that the same construction can be generalized to obtain instances with similar relaxation overhead  and diameter $O(n^{1/c})$ for larger constants $c < k$, by simply working with paths of length $n^{1/c}$ and repeating the path attachment step $c$ times.

\end{proof}

\begin{remark}
As suggested by the above examples, one might consider changing the priority function to preferentially select messages closer to the source. This can lead to improved work guarantees for the relaxed schedule. Indeed, we discuss one concrete example in Section~\ref{section:optimal-schedule}, where we show how to relax the optimal schedule on trees. 
However, it is not straightforward to construct such priority functions so that they also make sense on general graphs, which can have non-monotonic potentials and cycles.
\end{remark}

\paragraph{Discussion.} 
To interpret the results, first note that, in practice, the relaxation factor is in the order of $p$, the number of threads, and that $H$ is usually small (e.g., logarithmic) w.r.t. the total number of baseline updates $n$. 
Thus, in the good case, the $O( q^2 H )$ overhead can be seen as negligible: as $p$ iterations occur in parallel, the average time-to-completion should be $n / p + O( qH )$, which suggests almost perfect parallel speedup. 
At the same time, our worst-case instances shows that relaxed residual BP is not a ``silver bullet:'' there exist tree instances where it may lead to $\Omega(qn)$ message updates, i.e. asymptotically no speedup.
The next section shows experimentally that such worst-case instances are unlikely.


\section{Evaluation}\label{section:experiments}


We now empirically test the performance of the relaxed priority-based algorithms, comparing it against prior work.
For the experiments, we have implemented multiple priority-based algorithms and instantiated them with both exact and relaxed priority schedulers.

\subsection{Algorithms}

We implemented several variants of sequential belief propagation, including synchronous (round-robin), residual, weight decay, and residual without lookahead;
These variants are described in Section~\ref{section:preliminaries}. For residual splash, we implemented two variants. The first is the standard splash algorithm, as given in~\cite{gonzalez2012powergraph}. The second is our own optimized version we refer to as \emph{smart splash}, which only updates messages along breadth-first-search edges during a splash operation. This  variant has similar convergence as the baseline residual splash algorithm, but performs fewer message updates and should be more efficient.

We include the following instantiations of the algorithms in the benchmarks, and compare them against the sequential residual algorithm.

\paragraph{Exact schedulers.} We include the exact version of residual belief propagation (\textsf{Coarse-Grained Residual} or \textsf{CG}) and the version of splash algorithm (\textsf{Splash} or \textsf{S}) with the best value of $H$ ($2$ and $10$). For these algorithms, the scheduler is a standard concurrent priority queue. We omit here the results of the smart splash algorithm since it works significantly longer than under the relaxed scheduler described further. The other exact priority-based algorithms are not included, as they generally perform worse on our test instances.

\paragraph{Relaxed schedulers.} We compare the above algorithms against the algorithms we propose, i.e. relaxed versions of residual belief propagation (\textsf{Relaxed Residual}), weight decay belief propagation (\textsf{Weight-Decay}), residual without lookahead (\textsf{Priority}) and smart splash (\textsf{Relaxed Smart Splash}) with the best value $H = 2$. 
For all these algorithms, the scheduler is a Multiqueue with 4 priority queues per thread, as discussed in Section~\ref{section:relaxation}, which we found to work best (although other values behave similarly).

\paragraph{Other schedulers.} We ran many possible concurrent variants for the baseline algorithms, and 
chose the four best to we compare against our relaxed versions. 
First, we choose the parallel version of the standard synchronous belief propagation (\textsf{Synch}). 
We omit some synchronous algorithms such as the randomized synchronous belief propagation of Van der Merve et al.~\citep{merwe2019} since they perform consistently worse (see Appendix~\ref{app:random-synchronous}).

We also include the randomized version of splash algorithm (\textsf{RS}) proposed in the journal version of the paper~\cite{pmlr-v5-gonzalez09a} with $H = 2$, which performed best. 
This algorithm uses a similar idea of relaxation, but, crucially, instead of a Multiqueue scheduler, they implement a naive relaxed queue where threads randomly insert and delete into $p$ exact priority queues. While this distinction may seem small, it is known~\cite{AKLN17} that this variant does not implement a $k$-relaxed scheduler for any $k$, as its relaxation factor grows (diverges) as more and more operations are performed, and therefore corresponds to picking tasks at random to perform.  
As we will see in our experimental analysis, this does result in a significant difference between the number of additional (wasted) steps performed relative to a relaxed priority scheduler. Finally, we note that we are the first to implement this algorithm.

Finally, we include the algorithm proposed by Yin and Gao~\cite{10.1145/2661829.2662081}, hereafter referred to as the bucket algorithm (\texttt{Bucket}). At each round until convergence, this algorithm takes the $0.1 \cdot |V|$ best vertices according to the Splash metric, and updates the messages from these vertices.

\subsection{Models} \label{section:models}

We run our experiments on four Markov random fields models.

\paragraph{Trees.} As a simple base case, we consider a simple tree model similar to the analytical setting in Section~\ref{section:theory}. The underlying graph is a full binary tree, and the other parameters are set up as follows:
\begin{itemize}
    \item All variables are binary, i.e. the domain is $\{0, 1\}$ for each variable.
    \item Vertex factors are $(0.1, 0.9)$ for the root and $(0.5, 0.5)$ for all other vertices.
    \item Edge factors are $\psi_{ij}(x, y) = \begin{cases}1, & x = y \\ 0, & x \neq y\end{cases}$ for all edges.
\end{itemize}
As discussed in Section~\ref{section:theory}, these choices create a setup where the belief propagation has to propagate information from the root to all other nodes. Thus, under an optimal schedule, the total number of performed updates is equal to the number of edges. Since we know that all algorithms will converge on this model, we run the algorithms until exact convergence. 

\paragraph{Ising and Potts models.} Ising and Potts models are Markov random fields defined over an $n \times n$ grid graph, arising from applications in statistical physics. Both of Ising~\citep{elidan2006residual,knoll2015message} and Potts~\citep{Sutton:2007:IDS:3020488.3020534} models were used in prior work as test case, and in general they offer a class of good test instances, as they both exhibit complex cyclic propagations and are easy to generate.

For the parameters of the models, we mostly follow prior work in the setup.
For the Ising model, we select the factors similarly to \citep{elidan2006residual, knoll2015message}:
\begin{itemize}
    \item The variable domain is $\{-1, 1\}$ for all variables.
    \item Vertex factors are $\psi_i(x) = e^{\beta_i x}$.
    \item Edge factors are $\psi_{ij}(x, y) = e^{\alpha_{ij} x y}$.
    \item The parameters $\alpha_{ij}$ and $\beta_i$ are chosen uniformly at random from $[-1, 1]$.
\end{itemize}
For the Potts model, we select the factors following~\citep{Sutton:2007:IDS:3020488.3020534}:
\begin{itemize}
    \item The variable domain is $\{0, 1\}$ for all variables.
    \item Vertex factors are $\psi_i(x) = \begin{cases}e^{\beta_i}, & x = 1 \\ 1, & x = 0\end{cases}$.
    \item Edge factors are $\psi_{ij}(x, y) = \begin{cases}e^{\alpha_{ij}}, & x = y\\ 1, & x \neq y\end{cases}$.
    \item The parameters $\alpha_{ij}$ and $\beta_i$ are chosen uniformly at random from $[-2.5, 2.5]$.
\end{itemize}
For both Ising and Potts models, we set the convergence threshold to $10^{-5}$. That is, we terminate algorithm once all task have priority below this threshold.

\paragraph{LDPC codes.} Finally, we generate Markov random fields corresponding to the $(3, 6)$-LDPC (\emph{low density parity check code}~\citep{ldpc}) decoding. LDPC decoding is one of the more successful application of belief propagation. We consider a simple version of LDPC decoding task where convergence guarantees exist~\citep{richardson2001capacity}. However, we stress that coding theory is its own extensive research area, and far more optimized codes and decoding algorithms exist in practice---we simply use LDPC decoding to observe the comparative scaling behavior of our implementations on instances where synchronous belief propagation is guaranteed to converge. For a more detailed background on LDPC decoding and other aspects of coding theory, refer e.g.~to the book~\citep{richardson2008modern}.

More precisely, we consider $(3,6)$-LDPC decoding over a  binary symmetric channels. Informally, a $(3,6)$-LDPC code is a $(3,6)$-regular bipartite graph, where each degree $3$ node corresponds to a binary \emph{variable} and each degree $6$ node corresponds to a \emph{constraint} of form $x_{i_1} + x_{i_2} + \dotsc + x_{i_6} = 0$ over the neighboring variables $x_{i_1}, x_{i_2}, \dotsc, x_{i_6}$. Each sequence of variables that satisfies the all the constraints is \emph{codeword} of the code. The basic setup is then that we send a codeword over a \emph{channel} that flips each bit with probability $\varepsilon$, and the receiver will run belief propagation and use results of marginalization to infer the original codeword.

For our experiments, we build a $(3,6)$-LDPC instance with $2n$ variable nodes and $n$ constraint nodes by selecting a random $(3,6)$-regular bipartite graph, and initialize the node factors corresponding to the all-zero codeword sent over binary symmetric channel with error probability $\varepsilon = 0.07$. Under these conditions, belief propagation is guaranteed to correctly decode the instance with high probability~\citep{richardson2001capacity}; indeed, all the algorithms that converged decoded the codeword correctly in our experiments. The codeword length was again selected to get roughly comparable baseline running times as for the other instances.

Concretely, we get Markov random field where the underlying graph is a random bipartite graph with $3n$ nodes. For each variable node $i$, let $x_i \in \{ 0, 1 \}$ be the `transmitted' value of the variable, randomly generated to be $1$ with probability $\varepsilon$ and $0$ otherwise. The factors have the following structure:
    \begin{itemize}
        \item The domains of variable nodes are binary domains $\{0, 1\}$. For the constraint nodes, the domain is $\{ 0, 1 \}^6$---different bit masks of length $6$.
        \item The node factors for variable nodes are
        \[\psi_i(y) = \begin{cases}1 - \epsilon, & y = x_i \\ \epsilon, & y \ne x_i. \end{cases}\]
        For the constraint nodes, the node factor $\psi_c(y)$ is equal to the number of ones in $y \in \{ 0,1 \}^6$ modulo $2$; this effectively penalizes any value that does not satisfy the constraint.
        \item Edge factors $\psi_{i c}(x, y)$ is one if the corresponding bit in the $y \in \{ 0,1 \}^6$ equals $x \in \{ 0, 1 \}$, and is zero otherwise.
    \end{itemize}    
For the LDPC instances, we set the convergence threshold to $10^{-2}$ to ensure fast convergence; this approximates the behavior of actual LDPC decoders. 

\subsection{Methodology}

For each pair of algorithm and model, we run each experiment five times, and average the execution time and the number of performed updates on the messages. We executed on a 4-socket Intel Xeon Gold 6150 2.7 GHz server with four sockets, each with 18 cores,  and 512GB of RAM. The code is written in Java; we use Java 11.0.5 and OpenJDK VM 11.0.5. Our code is available at \href{https://github.com/IST-DASLab/relaxed-bp}{\footnotesize\textsf{https://github.com/IST-DASLab/relaxed-bp}}. 
Our code is fairly well optimized---in sequential executions it outperforms the C++-based open-source framework of libDAI~\cite{Mooij_libDAI_10} by more than 10x, and by more that 100x with multi-threading. 

To verify the accuracy of inference, we confirmed that all algorithms were able to recover the known correct codeword used to generate the LDPC code instance as expected. On the Ising model, all algorithms converged to almost identical marginals; on the Potts model, some non-relaxed algorithms had larger differences due to poor convergence.

\subsection{Result: moderate input sizes}

\begin{table}[t]
\centering
\resizebox{\linewidth}{!}{%
\begin{tabular}{lcccccccccccc}
\toprule
  & & \multicolumn{4}{c}{Prior Work} & \multicolumn{4}{c}{Relaxed} \\
  \cmidrule(lr){3-6}   \cmidrule(lr){7-10} 
Input & Residual & Synch & Coarse-G & Splash (10) & RS (2)  & Residual & Weight-Decay & Priority & Smart Splash (2)  \\
\midrule
Tree  & 1.30 min & 2.538x  & 0.265x & 1.648x      & 2.252x  & 1.391x  & 1.282x & 1.239x & 2.121x  \\
Ising & 2.76 min & 3.009x  & 0.801x & 5.393x      & 11.731x & 6.720x  & 6.276x & 5.759x & 14.175x \\
Potts & 3.02 min & ---     & 0.624x & 1.041x      & 11.855x & 7.454x  & 5.978x & 5.850x & 15.235x \\
LDPC  & 4.62 min & 17.735x & 1.166x & ---         & 5.150x  & 13.393x & 5.615x & ---    & 10.519x \\
\bottomrule
\end{tabular}
}

\caption{Algorithm speedups with respect to the sequential residual algorithm. Higher is better.}
\label{table:moderate:time}
\end{table}

\begin{table}[t]
\resizebox{\linewidth}{!}{%
\begin{tabular}{lcccccccccccc}
\toprule
  & & \multicolumn{4}{c}{Prior Work} & \multicolumn{4}{c}{Relaxed} \\
  \cmidrule(lr){3-6}   \cmidrule(lr){7-10} 
Input & Residual & Synch & Coarse-G & Splash (10) & RS (2)  & Residual & Weight-Decay & Priority & Smart Splash (2)  \\
\midrule
Tree  &10M   & 48.000x   & 1.003x   & 16.442x     & 8.344x  & 1.020x & 1.012x & 3.657x & 2.565x\\
Ising &25.3M & 45.006x   & 1.003x   & 9.266x      & 5.787x  & 1.058x & 1.068x & 1.816x & 1.878 \\
Potts &30M   & ---       & 1.006x   & 9.005x      & 5.983x  & 1.068x & 1.053x & 1.791x & 1.891x \\
LDPC  &7.23M & 4.404x    & 1.003x   & ---         & 4.089x  & 1.007x & 0.883x & --- & 0.973x \\
\bottomrule
\end{tabular}
}
\caption{Total updates relative to the sequential residual algorithm at 70 threads. Lower is better.}
\label{table:moderate:updates}
\end{table}

For basic performance tests, we run our experiments on four MRFs of moderate size: a binary tree of size $10^7$, an Ising model~\citep{elidan2006residual,knoll2015message} of size $10^3 \times 10^3$, a Potts~\citep{Sutton:2007:IDS:3020488.3020534} of size $10^3 \times 10^3$ and the decoding of $(3, 6)$-LDPC code~\citep{richardson2008modern} of size $3 \cdot 10^5$. The sizes of the inputs are chosen such that their execution is fairly long while the data still can fit into RAM.
We only present the most relevant baselines and best-performing algorithms here; see Appendix~\ref{app:experiments-moderate-all} for more algorithms.

We run the baseline residual belief propagation algorithm on one process since it is sequential, while all other algorithms are concurrent and, thus, are executed using $70$ threads. The execution times (speedups) relative to the sequential baseline are presented in Table~\ref{table:moderate:time}. Each cell of the table shows how much faster the corresponding algorithm works in comparison to the sequential residual one, i.e., higher is better; ``---'' means that the execution did not converge within a reasonable limit of time.

On trees, the fastest algorithm is, predictably, the synchronous one, since on tree-like models with small diameter $D$ it performs only approximately $O(D)$ times more updates in comparison to the sequential baseline, while being almost perfectly parallelizable. 
On Ising and Potts models, the best algorithm is Relaxed Smart Splash (RSS) with $H = 2$. The algorithm closest to it is Random Splash with $H = 2$, which is $20-30\%$ slower. For LDPC decoding, which is a locally tree-like model, the best-performing is again the Synchronous algorithm. 
We note the good performance of the relaxed residual algorithm, as well as of RSS, and the relatively poor performance of Random Splash, due to high numbers of wasted updates. 
Examining Table~\ref{table:moderate:updates}, we notice in general the relatively low number of wasted updates for relaxed algorithms. 
In summary, the choice of algorithm can depend on the model;  however, one may choose Relaxed Smart Splash since it performs well on all our models, while Relaxed Residual is fairly fast and simple to implement.

\subsection{Results: small input sizes}

In this subsection, we decrease the size of the inputs to enable more detailed scaling studies. Now, Tree model has $10^6$ vertices, Ising and Potts models are built on top of $300 \times 300$ grid graph, and, finally, LDPC model is set up with $30\,000$ length of the input vector. In general, we simply reduce the sizes of the models by approximately $10$.

\begin{figure*}[p]
\begin{center}
\includegraphics[width=\linewidth]{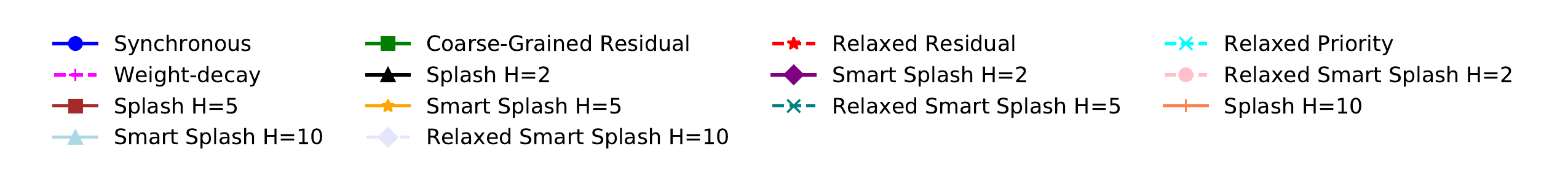}
\end{center}
\begin{adjustbox}{minipage=\linewidth}{
\begin{subfigure}{.5\textwidth}
  \centering
  \includegraphics[width=\textwidth]{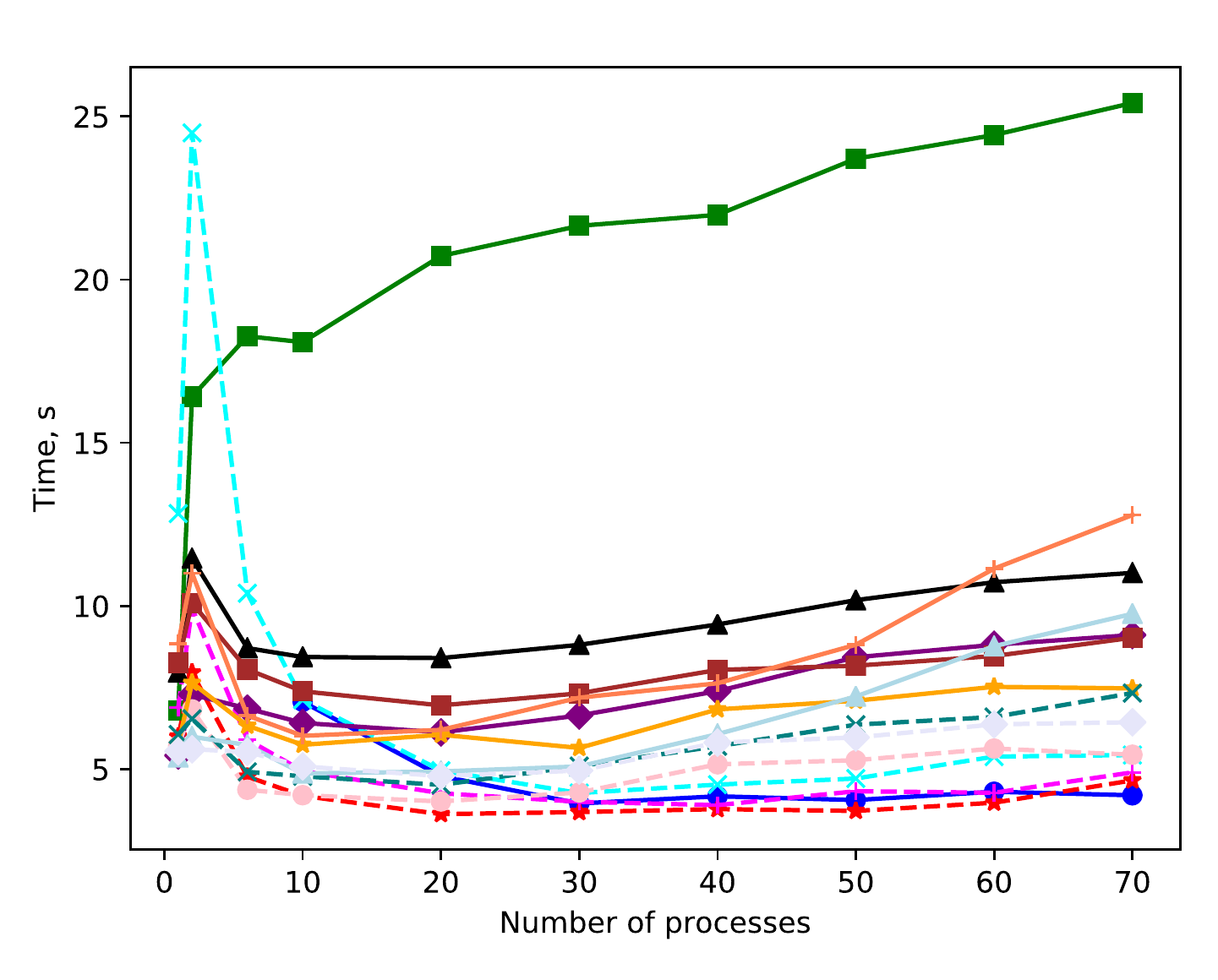}
  \vspace{-2em}
  \caption{Execution time}
  \label{fig:tree:time}
\end{subfigure}%
\begin{subfigure}{.5\textwidth}
  \centering
  \includegraphics[width=\textwidth]{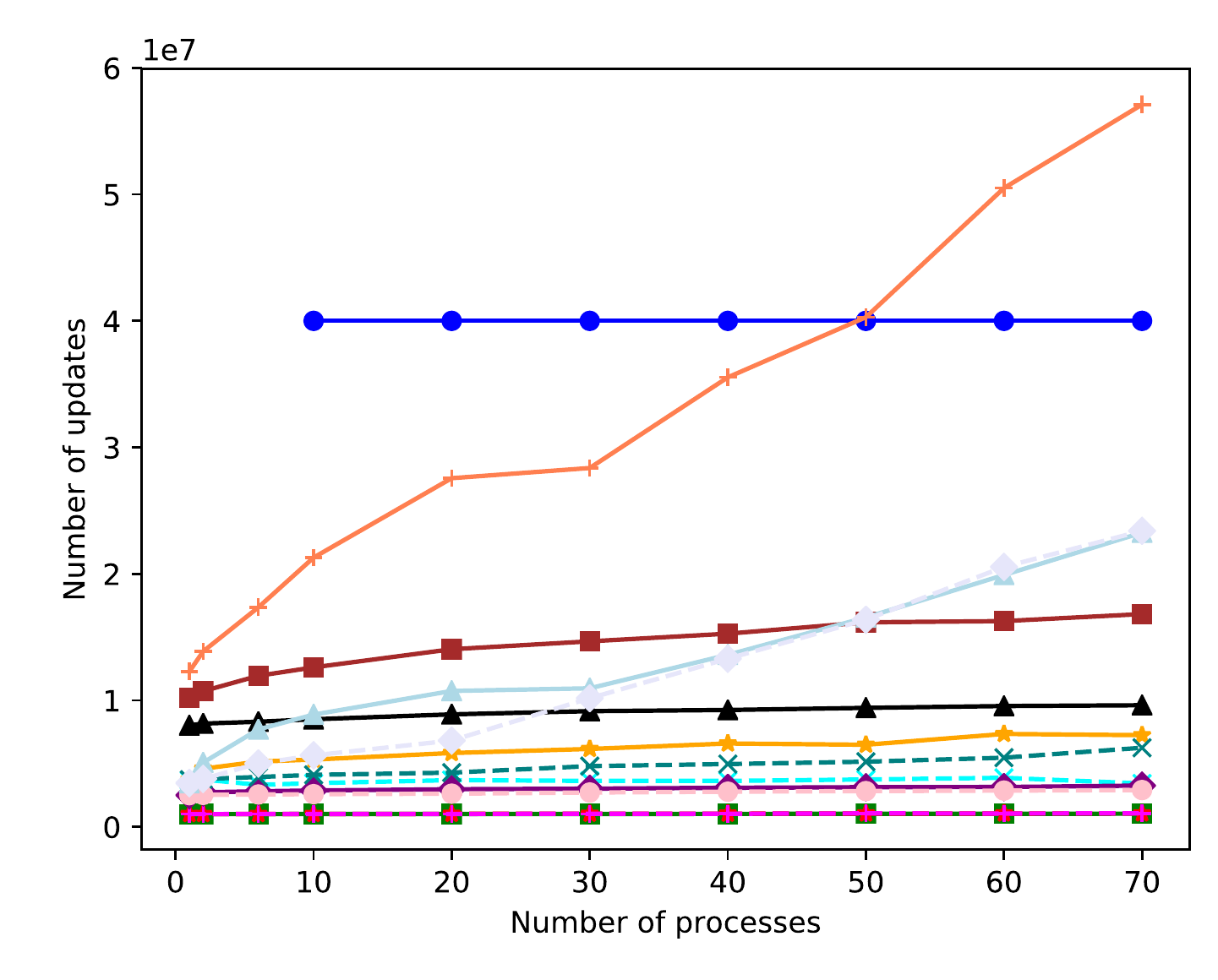}
  \vspace{-2em}
  \caption{Number of updates}
  \label{fig:tree:updates}
\end{subfigure}
\vspace{-0.2cm}
\caption{The results of the evaluation of the algorithms on the Tree model}
}
{
\begin{subfigure}{.5\textwidth}
  \centering
  \includegraphics[width=\textwidth]{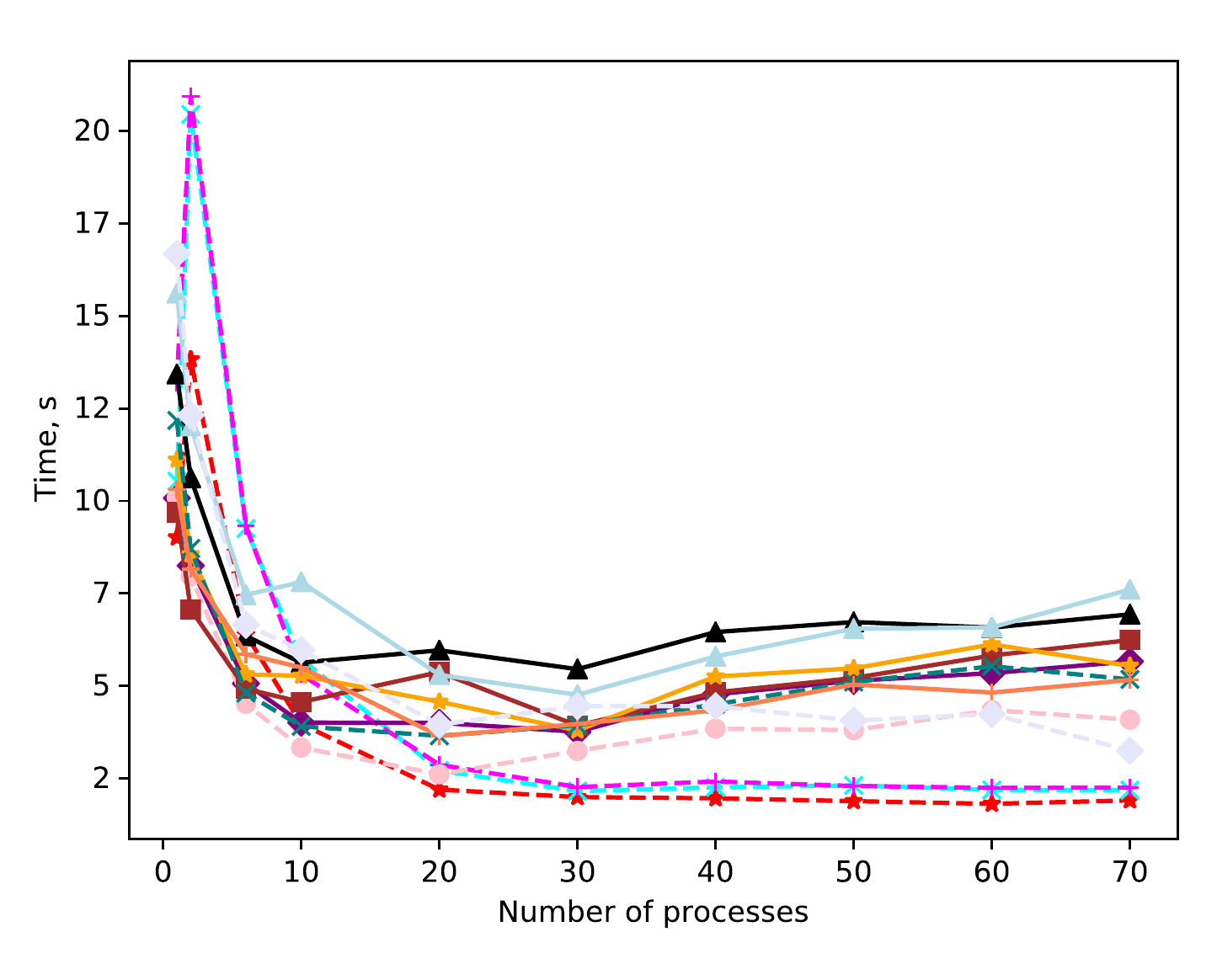}
  \vspace{-2em}
  \caption{Execution time}
  \label{fig:ising:time}
\end{subfigure}%
\begin{subfigure}{.5\textwidth}
  \centering
  \includegraphics[width=\textwidth]{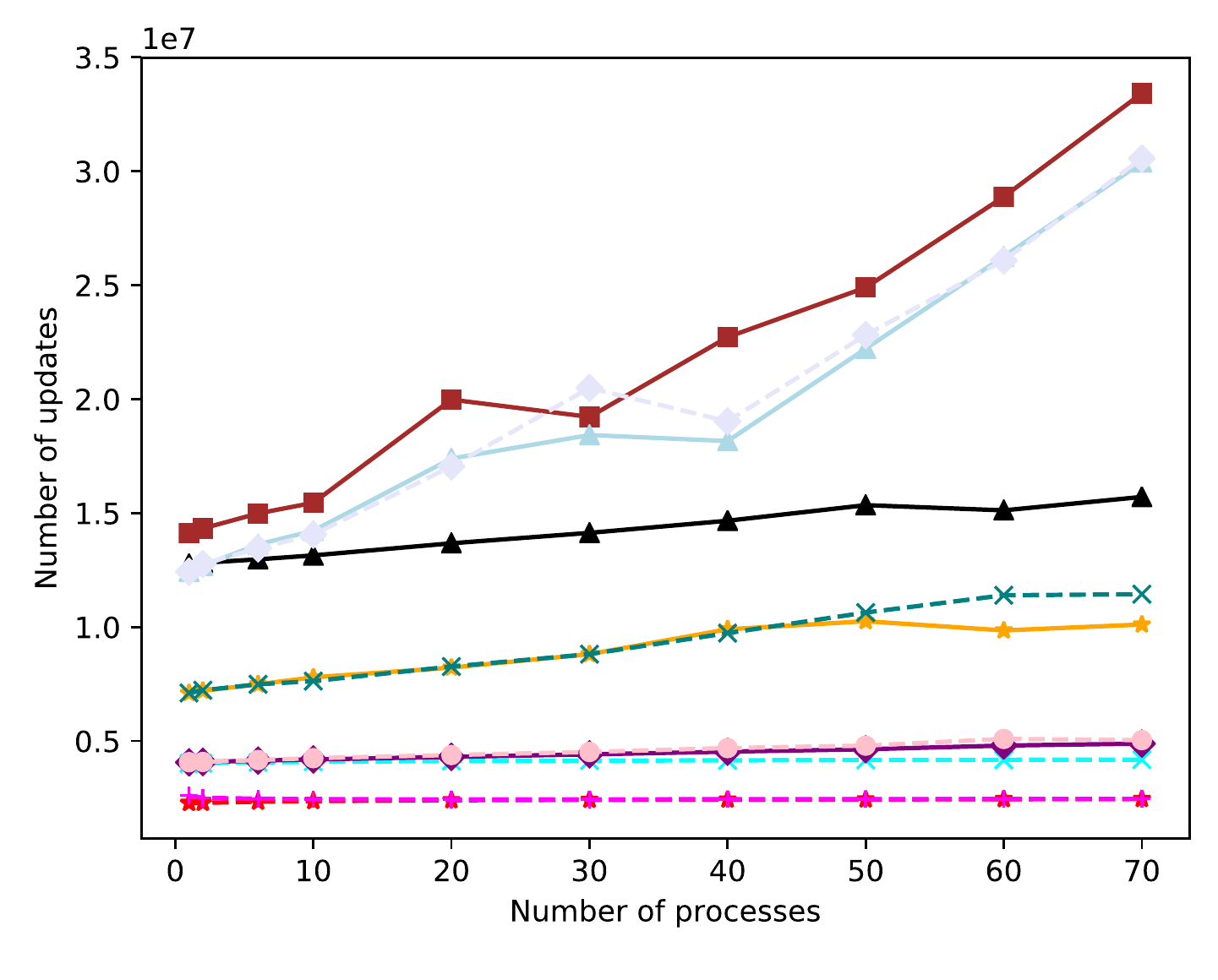}
  \vspace{-2em}
  \caption{Number of updates}
  \label{fig:ising:updates}
\end{subfigure}
\vspace{-0.2cm}
\caption{The results of the evaluation of the algorithms on Ising model}
}
\end{adjustbox}
\end{figure*}

\begin{figure*}[p]
\begin{center}
\includegraphics[width=\linewidth]{legend.pdf}
\end{center}

\begin{adjustbox}{minipage=\linewidth}{
\begin{subfigure}{.5\textwidth}
  \centering
  \includegraphics[width=\textwidth]{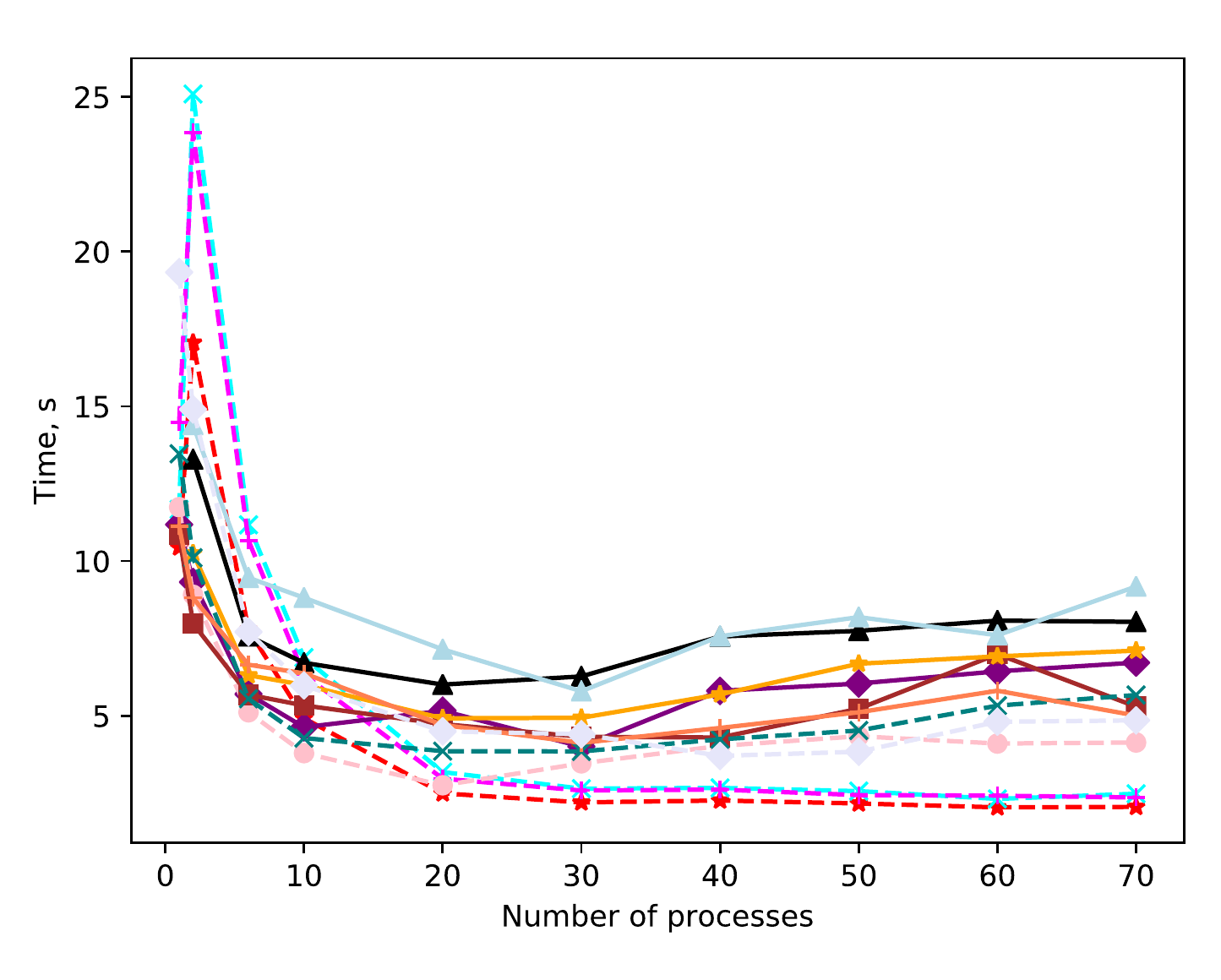}
  \vspace{-2em}
  \caption{Execution time}
  \label{fig:potts:time}
\end{subfigure}%
\begin{subfigure}{.5\textwidth}
  \centering
  \includegraphics[width=\textwidth]{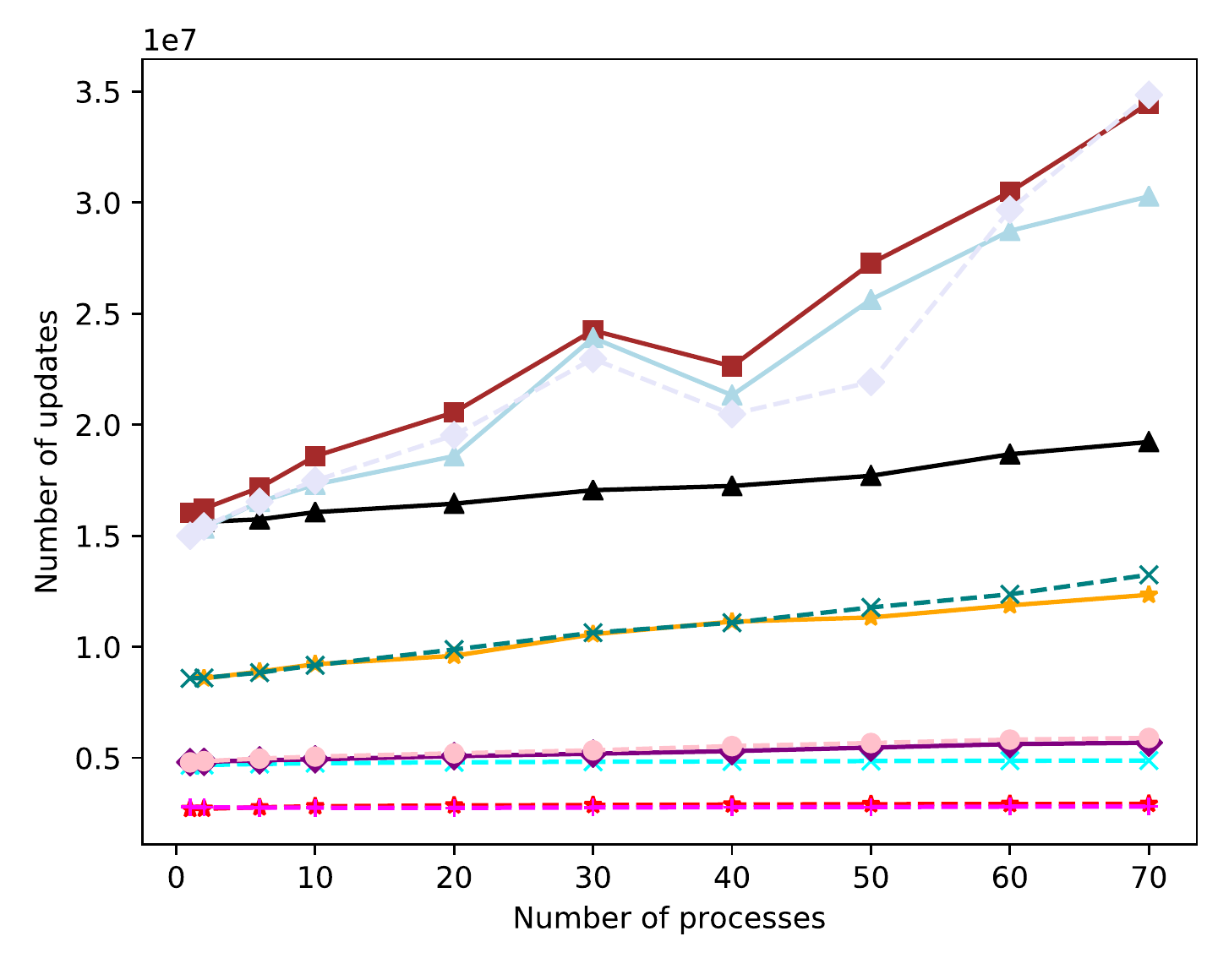}
  \vspace{-2em}
  \caption{Number of updates}
  \label{fig:legend}
\end{subfigure}
\vspace{-0.2cm}
\caption{The results of the evaluation of the algorithms on Potts model}
}
{
\begin{subfigure}{.5\textwidth}
  \centering
  \includegraphics[width=\textwidth]{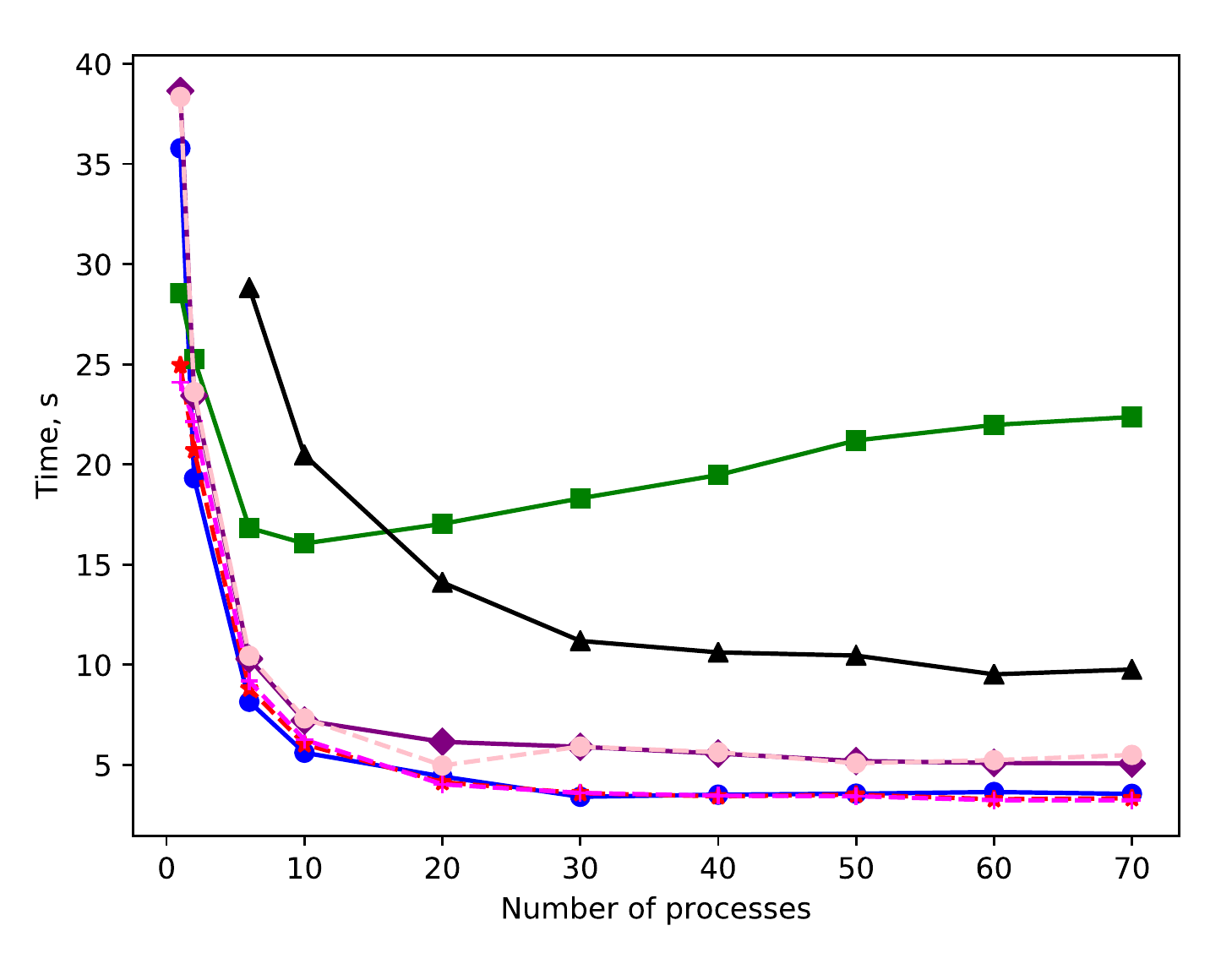}
  \vspace{-2em}
  \caption{Execution time}
  \label{fig:ldpc:time}
\end{subfigure}%
\begin{subfigure}{.5\textwidth}
  \centering
  \includegraphics[width=\textwidth]{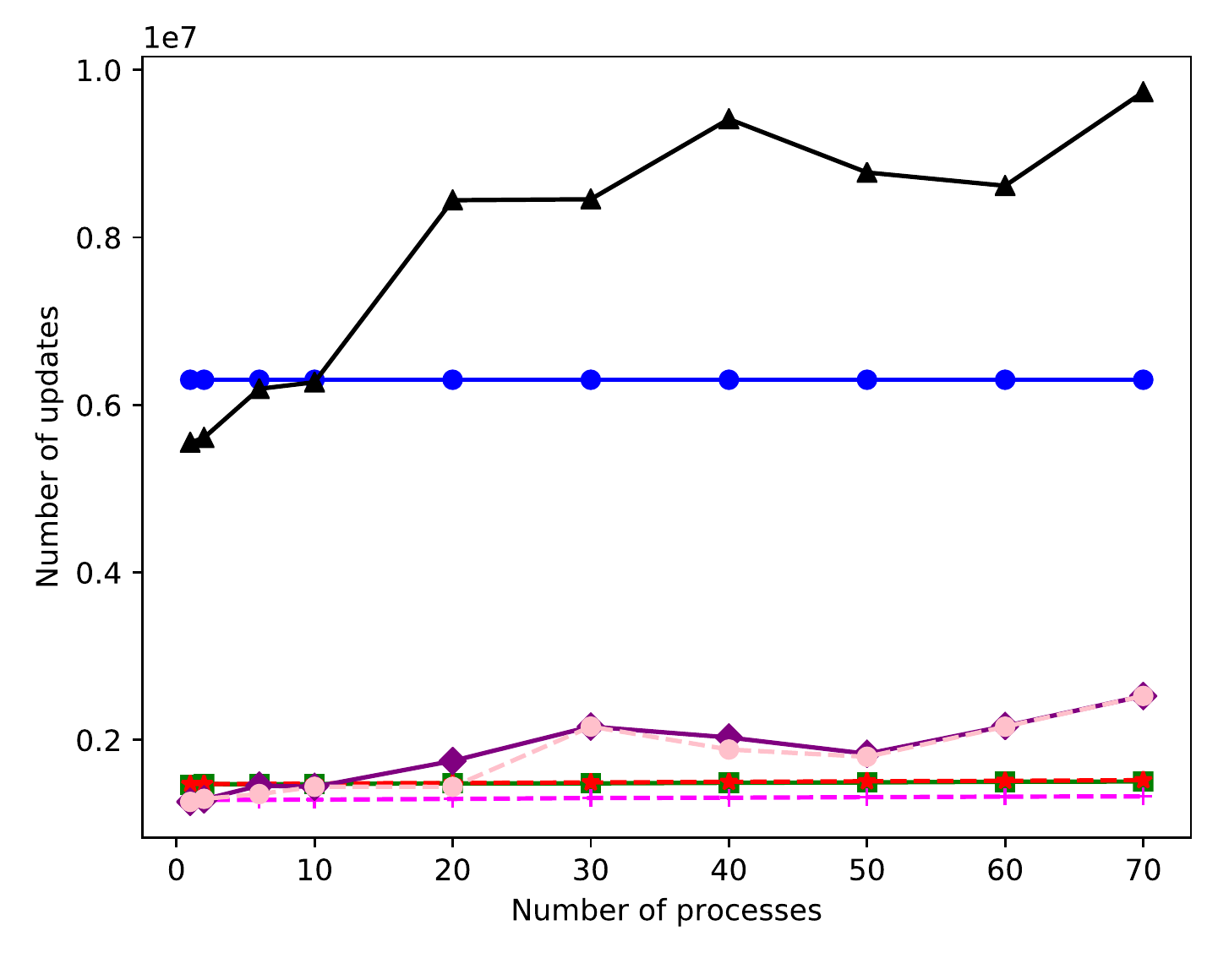}
  \vspace{-2em}
  \caption{Number of updates}
  \label{fig:ldpc:updates}
\end{subfigure}
\vspace{-0.2cm}
\caption{The results of the evaluation of the algorithms on decoding LDPC code}
}
\end{adjustbox}
\end{figure*}

\paragraph{How to read the plots.} There are two types of plots per each model: the first shows the execution time of the algorithms, while the other one shows the number of updates performed. 
On the $x$ axis we have the number of threads the algorithms were run on, while on the $y$ axis we have: the time in seconds (for time plots) and the number of updates (for update plots). The dashed lines on the plots correspond to the algorithms that use a relaxed scheduler, while the others use either no concurrent scheduler, or an exact priority queue.

Whenever we have omitted algorithms from the plots or display incomplete data, this indicates poor performance for that algorithm on the metric displayed on the graph: either the algorithm did not converge or the values exceed the limit of the plot.

\paragraph{Tree model.}
As one can observe on the time plot (Figure~\ref{fig:tree:time}), the three algorithms with the best scaling on the tree instance are the synchronous belief propagation, relaxed residual and the weight-decay algorithm. For the relaxed algorithms, this mirrors our theoretical analysis from Section~\ref{section:theory}: as can be seen from Figure~\ref{fig:tree:updates}, the relaxation incurs very low overhead in terms of additional updates, while the overhead from parallelization is also low. By contrast, the exact residual belief propagation performs exactly the minimum number of updates needed, but scales very badly due to the contention on the priority queue.

We note that on the tree instance, the synchronous belief propagation also scales very well when parallelized. The amount of work can be split evenly between the threads, and only $O(\log n)$ synchronous rounds are required for convergence. 

\paragraph{Ising and Potts model.}
Ising and Potts models represent more challenging instances with lots of cycles, and are generally thought to be more representative of hard general graph instances for belief propagation. As can be seen in Figures~\ref{fig:ising:time} and~\ref{fig:potts:time}, relaxed algorithms perform consistently well on these instances, with relaxed residual belief propagation giving consistently the fastest convergence. These are followed by the exact splash algorithms, which generally perform slightly worse; however, the scaling seems to be somewhat sensitive to the choice of the parameter $H$. Both the synchronous and exact residual belief propagation are omitted, as the former did not consistently converge, and the latter was very slow.

An interesting insight is that the exact variants of splash and smart splash do not converge at all in single-threaded executions for some values of the parameter $H$, but always converge on two and more threads. Similarly, synchronous belief propagation, which has a fixed schedule, does not converge. By contrast, relaxed smart splash converged under all parameter values. We conjecture that this is due to the phenomenon observed by~\citep{knoll2015message}: exact priority-based algorithms may get stuck in non-convergent cyclic schedules, and injecting randomness into the schedule may help the algorithm to `escape' these situations. In particular, relaxation to the priority queue, i.e., sometimes executing low-priority items, can provide a such source of randomness. Similarly, an increase in the number of threads leads to the relaxation of the algorithm even for exact schedulers, as several messages are processed in parallel, not only the best one. Thus, we empirically observe that the randomness in the relaxation might help belief propagation to avoid bad cyclic schedules and, therefore, converge.

\paragraph{LDPC model.}
There are five algorithms that perform similarly (Figure~\ref{fig:ldpc:time}): synchronous belief propagation, relaxed residual belief propagation, the weight decay algorithm, relaxed smart splash with $H = 2$ and, finally, smart splash with $H = 2$. The other algorithms did not converge within our five minutes time limit per experiment. 

We note that synchronous belief propagation performs very well on this instance. This is not surprising, as standard belief propagation is known to perform well in LDPC decoding. Generally speaking, the necessary propagation chains seem to be very short on LDPC instances, and the synchronous algorithm parallelizes well in such cases.

\begin{table}
\begin{center}
{
\begin{tabular}{llcccc}
\toprule
 && \multicolumn{4}{c}{Message updates}\\
\cmidrule(lr){3-6}
                 & Threads & Tree    & Ising   & Potts   & LDPC    \\
\midrule
Exact  & 1       & 1000000 & 2279000 & 2700000 & 1464000 \\
\midrule
Relaxed & 1 & +0.14\%  & +0.11\% & -0.01\% & +0.55\% \\
                 & 2 & +0.26\%  & +0.24\% & +0.37\% & +0.57\% \\
                 & 6 & +0.56\%  & +2.50\% & +2.70\% & +0.64\% \\
                 & 10 & +0.92\% & +3.71\% & +4.45\% & +1.05\% \\
                 & 20 & +2.08\% & +5.27\% & +5.87\% & +1.41\% \\
                 & 30 & +2.90\% & +6.10\% & +6.56\% & +1.87\% \\
                 & 40 & +3.48\% & +6.52\% & +7.39\% & +2.35\% \\
                 & 50 & +5.04\% & +6.83\% & +7.92\% & +2.83\% \\
                 & 60 & +4.96\% & +7.39\% & +8.28\% & +3.20\% \\
                 & 70 & +5.74\% & +7.71\% & +8.53\% & +3.70\% \\
\bottomrule
\end{tabular}
}
\captionof{table}{Number of additional message updates performed by relaxed residual belief propagation compared to exact residual belief propagation.}
\label{table:effect}
\end{center}
\end{table}

\subsection{Results: the effects of relaxation}

\paragraph{Relaxation and the number of updates.} In Table~\ref{table:effect}, we measure how many more updates the relaxed residual algorithm needs to perform in comparison to the number of updates performed by the standard sequential residual algorithm, denoted as ``baseline''. We count the total number of updates only approximately: we check the convergence condition only after every $1000$ iterations. Instance sizes are the same as in the scaling experiments, i.e. `small'. 

The left column indicates whether it is a baseline algorithm or the number of threads for relaxed residual belief propagation. The other columns present the numbers for each model we consider. Each cell contains the corresponding number of updates and how many more updates the relaxed version of the algorithm executed (percentage).

On one process, relaxed residual performs more updates than the baseline does, except in the case of the Potts model. It is expected since our algorithm uses relaxed Multiqueue instead of the strict priority queue. Moreover, as expected the overhead on the number of updates in comparison to the baseline increases with the number of threads. This is again due to the relaxation of the priority queue--recall that we allocate $4\times$ more queues than threads.  Interestingly, this overhead is limited even on $70$ threads---its maximum value is $9\%$ maximum. This explains the good performance of our algorithm: we reduce the contention by relaxing accesses to the priority queue, while at the same time the total number of updates does not increase significantly.

\begin{table}
\centering
\begin{tabular}{lcccc}
\toprule
        & \multicolumn{4}{c}{Speedup}\\
\cmidrule(lr){2-5}
Threads & Tree & Ising & Potts & LDPC\\
\midrule
1 & 0.89x & \textbf{1.08x} & \textbf{1.04x} & \textbf{1.14x}\\
2 & 0.75x & 0.51x & 0.47x & \textbf{1.13x}\\
6 & \textbf{1.20x} & 0.77x & 0.73x & \textbf{1.17x}\\
10 & \textbf{1.16x} & \textbf{1.01x} & 0.94x & \textbf{1.20x}\\
20 & \textbf{1.36x} & \textbf{1.66x} & \textbf{1.89x} & \textbf{1.49x}\\
30 & \textbf{1.38x} & \textbf{1.88x} & \textbf{1.82x} & \textbf{1.65x}\\
40 & \textbf{1.61x} & \textbf{2.21x} & \textbf{1.90x} & \textbf{1.62x}\\
50 & \textbf{1.91x} & \textbf{2.67x} & \textbf{2.36x} & \textbf{1.48x}\\
60 & \textbf{1.89x} & \textbf{2.66x} & \textbf{2.85x} & \textbf{1.55x}\\
70 & \textbf{1.61x} & \textbf{2.71x} & \textbf{2.44x} & \textbf{1.52x}\\
\bottomrule
\end{tabular}
\caption{Speedup of relaxed residual belief propagation versus the best non-relaxed alternative on different thread counts. We note that overhead of parallelization can overcome the benefits on small thread counts, as seen in the scaling experiments.} 
\label{table:speed}
\end{table}

\paragraph{Relaxed versus non-relaxed algorithms.} In Table~\ref{table:speed}, we analyze the speedups obtained by the relaxed residual algorithm relative to the best-performing non-relaxed alternative across models and thread counts. Instance sizes are the same as in the scaling experiments, i.e. `small'. 

We notice that our algorithm outperforms the alternatives in most of the cases, often by a large margin---the highest speedup is of $2.85\times$, whereas the highest slow-down is of $0.47x$. Both occur on the Potts model, which is generally the most difficult instance in our tests.
Overall, the combination of our relaxed scheduling framework combined with the standard residual belief propagation is clearly the algorithm of choice at high thread counts, where it consistently outperforms the alternatives; on the other hand, relaxed residual also performs reasonably well on a single thread, making it a consistently good choice all across the board.

\section{Discussion}

We have investigated the use of relaxed schedulers in the context of the classic belief propagation algorithm for inference on graphical model, and have shown that this approach leads to an efficient family of algorithms, which improve upon the previous state-of-the-art non-relaxed parallelization approaches in our experiments. Overall, our relaxed implementations, either Relaxed Residual or Relaxed Smart Splash, have state-of-the-art performance in multithreaded regimes, making them a good generic choice for any belief propagation task.

For future work, we highlight two possible directions. First is to extend our theoretical analysis to cover more types of instances; however, as we have seen, the structure of belief propagation schedules can be quite complicated, and the challenge is the figure out a proper framework for more general analysis. Second possible direction is extending our empirical study to a massively-parallel, multi-machine setting. 

\subsection*{Broader Impact}
As this work is focused on speeding up existing inference techniques and does not focus on a specific application, the main benefit is enabling belief propagation applications to process data sets more efficiently, or enable use of larger data sets. We do not expect direct negative societal consequences to follow from our work, though we note that as with all heuristic machine learning techniques, there is an inherent risk of misinterpreting the results or ignoring biases in the data if proper care is not taken in application of the methods. However, such risks exist regardless of our work.

\subsection*{Acknowledgements}
We thank Marco Mondelli for discussions related to LDPC decoding, and Giorgi Nadiradze for discussions on analysis of relaxed schedulers. This project has received funding from the European Research Council (ERC) under the European Union's Horizon 2020 research and innovation programme (grant agreement No 805223 ScaleML) and from Government of Russian Federation (Grant 08-08).

\bibliographystyle{plainnat}
\bibliography{bp}

\appendix



\section{Optimal schedule on trees}\label{section:optimal-schedule}


On trees, the belief propagation gives exact marginals under any schedule that updates each edge infinitely often. However, there is an optimal schedule that updates each message exactly once, requiring $O(n)$ message updates~\citep{bp}. Assume the tree has a fixed root $v$:
\begin{enumerate}
    \item In the first phase, all messages towards the root are updated starting from the leaves; each message is updated only after all its predecessors have been updated.
    \item In the second phase, all messages away from the root are update starting from the root.
\end{enumerate}
This schedule can be modeled in the priority-based scheduling framework as follows:
\begin{enumerate}
    \item Initially, the outgoing messages at leaf nodes have priority $n$, and all other messages have priority $0$.
    \item When message is updated with non-zero priority, its priority is changed to $0$.
    \item Once all messages $\mu_{k \to i}$ for $k \in N(i) \setminus \{ j \}$ have been updated once with non-zero priority, the message $\mu_{i \to j}$ changes to priority to minimum of update priorities of the incoming edges minus one.
\end{enumerate}
This priority function can clearly be implemented by keeping a constant amount of extra information per message. When the above schedule is executed with an exact scheduler, the algorithm will update each message once with non-zero priority before considering any messages with zero priority, and by following the analysis of~\citep{bp}, one can see that the algorithm has converged at that point.

Similarly, in the relaxed version of the schedule, the algorithm has converged once all messages have been updated once with non-zero priority. In addition, some messages may be updated multiple times with priority~$0$; we call these \emph{wasted} updates, and the updates done while the message has non-zero priority \emph{useful} updates.

\begin{claim}
The relaxed version of the optimal schedule on trees performs $O( n + k^2 H)$ message updates, where $H$ is the height of the tree.
\end{claim}

\begin{proof}
For the purposes of analysis, assign messages into buckets $B_1, B_2 \dotsc, B_n$ so that bucket $B_\ell$ contains the messages that will have their useful update done with priority $\ell$. One can observe that the update priority of message $\mu_{i \to j}$ is the $n - d$, where $d$ is the maximum distance from node $i$ to a leaf using a path that does not cross edge $\{ i, j \}$. Since this is bounded by the diameter of the tree, there are at most $2H$ non-empty buckets.

Assume that all messages in buckets $B_n, B_{n-1}, \dotsc, B_{\ell + 1}$ have been already had a useful update. We now show that in there can be at most $k^2$ wasted updates before all messages in $B_\ell$ have had a useful update. Since all earlier buckets have been processed, all messages in $B_\ell$ have either already had a useful update, or have priority~$\ell$. Let $b$ be the number of messages remaining in bucket $B_\ell$:
\begin{itemize}
    \item While $b \ge k$, there are at least $k$ messages with priority $\ell$, so each update is a useful update of a message in $B_\ell$
    \item When $b < k$, there can be wasted updates. However, since buckets $B_n, B_{n-1}, \dotsc, B_{\ell + 1}$ have had all useful updates, the top elements in the schedule will be from bucket $B_\ell$, and thus by the guarantees of the scheduler, there can be at most $k-1$ wasted updates before the top element is processed. Thus, in $b(k-1) = O(k^2)$ updates, all remaining messages of $B_\ell$ will have their useful update.
\end{itemize}
By an inductive argument, all non-empty buckets have been processed after $O(k^2 H)$ wasted updates, so the total number of updates is $O(n + k^2 H)$.
\end{proof}


\section{Additional experiments}\label{app:experiments}


\subsection{Results: moderate input sizes}\label{app:experiments-moderate-all}

In Tables~\ref{table:moderate:full:time} and~\ref{table:moderate:full:updates}, we present running times and update counts for all algorithms included in our main experiments.
To present all the executed algorithms in the table, we shrink the abbreviations a little bit: \textsf{Coarse-Grained} now becomes \textsf{CG}, \textsf{Splash} becomes \textsf{S}, \textsf{Random Splash} becomes \textsf{RS},
\textsf{Relaxed Residual} rests \textsf{Residual}, \textsf{Weight-Decay} becomes \textsf{WD}, \textsf{Relaxed Priority} becomes \textsf{Priority}, and, finally, \textsf{Relaxed Smart Splash} becomes \textsf{RSS}.

\begin{landscape}
\begin{table}
\centering
\begin{tabular}{lccccccccccccc}
\toprule
 & & \multicolumn{7}{c}{Prior Work} & \multicolumn{5}{c}{Relaxed} \\
\cmidrule(lr){3-9}   \cmidrule(lr){10-14} 
Input & Residual & Synch   & CG     & S $2$  & S $10$ & RS $2$  & RS $10$ & Bucket & Residual & WD & Priority & RSS $2$ & RSS $10$  \\
\midrule
Tree  & 1.30 min & 2.538x  & 0.265x & 0.608x & 1.648x & 2.252x  & 2.241x  & 1.692x & 1.391x   & 1.282x & 1.239x & 2.121x  & 2.110x \\
Ising & 2.76 min & 3.009x  & 0.801x & 0.609x & 5.393x & 11.731x & 13.512x & 3.311x & 6.720x  & 6.276x & 5.759x & 14.175x & 10.337x\\
Potts & 3.02 min & ---     & 0.624x & 0.484x & 1.041x & 11.855x & 12.854x & ---    & 7.454x  & 5.978x & 5.850x & 15.235x & 11.091x \\
LDPC  & 4.62 min & 17.735x & 1.166x & ---    & ---    & 5.150x  & ---     & 3.044x & 13.393x & 5.615x & ---    & 10.519x & ---\\
\bottomrule
\end{tabular}

\caption{Algorithm speedups with respect to the sequential residual algorithm. Higher is better.}
\label{table:moderate:full:time}
\end{table}

\begin{table}
\centering

\begin{tabular}{lccccccccccccc}
\toprule
 & & \multicolumn{7}{c}{Prior Work} & \multicolumn{5}{c}{Relaxed} \\
\cmidrule(lr){3-9}   \cmidrule(lr){10-14} 
Input & Residual & Synch   & CG     & S $2$  & S $10$  & RS $2$ & RS $10$ & Bucket & Residual & WD     & Priority & RSS $2$ & RSS $10$  \\
\midrule
Tree  & 10M      & 48.000x & 1.003x & 8.658x & 16.442x & 8.344x & 15.197x &5.110x  & 1.020x   & 1.012x & 3.657x   & 2.565x  & 5.027x \\
Ising & 25.3M    & 45.006x & 1.003x & 5.719x & 9.266x  & 5.787x & 10.232x &3.996x  & 1.058x   & 1.068x & 1.816x   & 1.878x  & 6.147x \\
Potts & 30M      & ---     & 1.006x & 5.903x & 9.005x  & 5.983x & 9.109x  &---     & 1.068x   & 1.053x & 1.791x   & 1.891x  & 6.328x \\
LDPC  & 7.23M    & 4.404x  & 1.003x & ---    & ---     & 4.089x & ---     &1.668x  & 1.007x   & 0.883x & ---      & 0.973x  & --- \\
\bottomrule
\end{tabular}

\caption{Total updates relative to the sequential residual algorithm at 70 threads. Lower is better.}
\label{table:moderate:full:updates}
\end{table}
\end{landscape}

Table~\ref{table:moderate:full:time} contains the execution times (speedups) of the algorithms relative to the sequential baseline. Table~\ref{table:moderate:full:updates} contains the number of updates performed by the algorithms in compare to the number of updates performed by the sequential baseline. The results do not differ much from the ones presented in the main body of the paper. The only notable thing is that \textsf{Random Splash} with $H = 10$ it performs better on Ising and Potts model than \textsf{Random Splash} with $H = 2$. However, we chose \textsf{Random Splash} with $H = 2$ as the best one, since Random Splash with $H = 10$ does not finish on LDPC input. Nevertheless, \texttt{Relaxed Smart Splash} outperforms \textsf{Random Splash} with both settings of $H$.

\subsection{Random synchronous algorithm}\label{app:random-synchronous}

In Table~\ref{tab:random-synch}, we present the execution time of random synchronous algorithm of Van der Merve et al.~\citep{merwe2019} on $70$ threads (\textsf{Random Synch} $70$) with different values of $lowP = 0.1, 0.4 \text{ and } 0.7$, where the parameter $lowP$ controls the random selection fraction $p$ in steps where the algorithm is converging slowly (see Section~\ref{app:parallel-bp}).
We compare it with the execution time of two baselines: Synchronous algorithm on $70$ threads (\textsf{Synch} $70$) and Relaxed Residual on one process (\textsf{RR} $1$). Cells with `---' indicate executions that either take more than five minutes to run or simply do not converge.

To summarize, we did not include the execution time of random synchronous algorithm in the scaling plots since it exceeds the execution time of one of the baselines in all cases.

\begin{table}[h!]
\begin{center}
\begin{tabular}{ll@{\ \ \ \ \ }rrrr}
\toprule
 && \multicolumn{4}{c}{Running time (s)}\\
\cmidrule(lr){3-6}
\multicolumn{2}{l}{Algorithm}   & Tree    & Ising   & Potts   & LDPC    \\
\midrule
Synch 70      &  & 4.088 &  --- & --- & 3.504 \\
RR 1      &  & 5.579 & 9.012 & 10.583 & 25.663 \\
\midrule
Random Synch 70 & $lowP = 0.1$ & 37.052 & 62.629 & --- & 28.543 \\
      & $lowP = 0.4$ & 8.420 & 20.396 & --- & 7.269 \\
      & $lowP = 0.7$ & 6.306 & 12.581  & ---  & 4.791 \\
\bottomrule
\end{tabular}
\end{center}
\caption{Randomized synchronous algorithm versus baselines.}
\label{tab:random-synch}
\end{table}

\end{document}